%% file: 0main.tex
\documentclass[runningheads]{llncs}

\usepackage{graphicx}

\usepackage{times}
\usepackage{soul}
\usepackage{url}
\usepackage{hyperref}

\usepackage[utf8]{inputenc}
\usepackage[small]{caption}
\usepackage{graphicx}

\usepackage{amsmath}

\usepackage{booktabs}
\usepackage{algorithm}
\usepackage{arevmath}
\usepackage[noend]{algpseudocode}
\usepackage{listings,multicol}
\usepackage{multicol}
\usepackage{longtable}
\usepackage{oz}
\usepackage{listings}
\usepackage{amsfonts}
\usepackage{relsize}
\usepackage{urwchancal}
\usepackage{stmaryrd}
\usepackage{ dsfont }
\usepackage{placeins}
\usepackage{longtable}
\usepackage{standalone}
\usepackage[htt]{hyphenat}
\usepackage{float}
\usepackage{mathtools}

\newcommand{\axiom}[0]{\ensuremath{\alpha}}
\newcommand{\axiomatisation}[1]{\ensuremath{\axiom(#1)}}

\newcommand{\sh}[1]{\ensuremath{\texttt{sh}{:}{\texttt{#1}}}}
\newcommand{\rdf}[1]{\ensuremath{\texttt{rdf}{:}{\texttt{#1}}}}

\newcommand{\trip}[3]{\ensuremath{\langle #1,\allowbreak #2,\allowbreak #3 \rangle}}
\newcommand{\pair}[2]{\ensuremath{\text{:}\langle #1,\allowbreak #2 \rangle}}

\newcommand{\gramdef}[0]{\ensuremath{:=}}
\newcommand{\grameq}[0]{\ensuremath{\doteq}}

\newcommand{\vn}[0]{\ensuremath{n}}
\newcommand{\vx}[0]{\ensuremath{x}}
\newcommand{\vy}[0]{\ensuremath{y}}
\newcommand{\vz}[0]{\ensuremath{z}}
\newcommand{\conc}[0]{\ensuremath{\texttt{c}}}

\newcommand{\structureOmegaMain}[0]{\ensuremath{I}}
\newcommand{\structureOmegaOne}[0]{\ensuremath{\Omega_{G,\vShapeDocument_1}}}
\newcommand{\structureOmegaTwo}[0]{\ensuremath{\Omega_{G,\vShapeDocument_2}}}
\newcommand{\structureOmegaG}[0]{\ensuremath{\Omega_{G}}}
\newcommand{\structureOmegaS}[0]{\ensuremath{\Omega_{G,\vShapeDocument}}}

\newcommand{\transeq}[0]{\ensuremath{\; \doteq \;}}
\newcommand{\tr}[0]{\ensuremath{r}}
\newcommand{\constantList}[0]{\ensuremath{C}}

\newcommand{\vShapeDocument}[0]{\ensuremath{M}}
\newcommand{\vShapes}[0]{\ensuremath{\bar{S}}}
\newcommand{\vshape}[0]{\texttt{s}}
\newcommand{\vshapet}[0]{\ensuremath{t}}
\newcommand{\vshapec}[0]{\ensuremath{d}}

\newcommand{\vs}[0]{\ensuremath{s}}

\newcommand{\namedshapedefinition}[0]{referenced shape definition}
\newcommand{\namedshape}[0]{referenced shape}

\newcommand{\sconst}{\texttt{s}}
\newcommand{\lO}{\texttt{O}}
\renewcommand{\S}{\texttt{S}}
\newcommand{\Z}{\texttt{Z}}
\newcommand{\A}{\texttt{A}}
\newcommand{\T}{\texttt{T}}
\newcommand{\D}{\texttt{D}}
\newcommand{\E}{\texttt{E}}

\renewcommand{\O}{\texttt{O}}
\newcommand{\C}{\texttt{C}}

\newcommand{\X}{$\varnothing$}

\newcommand{\hasshape}[2]{\ensuremath{\texttt{hasShape}(#1,#2)}\xspace}
\newcommand{\hasshapePredicate}[0]{\ensuremath{\texttt{hasShape}}\xspace}

\newcommand{\SCL}[0]{\ensuremath{\texttt{\textbf{SCL}}}}

\newcommand{\Diagonal}[0]{\ensuremath{D}}

\newcommand{\taus}[2]{\ensuremath{\tau_{#1}(#2)}}

\newcommand{\isA}[0]{\texttt{isA}\xspace}

\newcommand{\isIRI}[0]{\text{IRI}}
\newcommand{\isLiteral}[0]{\text{literal}}
\newcommand{\isBlank}[0]{\text{blank}}
\newcommand{\hasdatatype}[0]{\text{dt}}

\usepackage{xcolor}
\usepackage{listings}
\usepackage{lipsum}

\usepackage{listings}
\usepackage{lstautogobble}
\usepackage{color}
\usepackage{zi4}

\definecolor{bluekeywords}{rgb}{0.13, 0.13, 1}
\definecolor{greencomments}{rgb}{0, 0.5, 0}
\definecolor{redstrings}{rgb}{0.9, 0, 0}
\definecolor{graynumbers}{rgb}{0.5, 0.5, 0.5}

\usepackage[utf8]{inputenc}
\usepackage[T1]{fontenc}
\usepackage{zlmtt}
\usepackage{fancyvrb}
\usepackage{caption}

\usepackage{listings}
  \newcommand{\TileSet}[0]{\ensuremath{T}}
  \newcommand{\HorizontalRel}[0]{\ensuremath{\mathsf{H}}}
  \newcommand{\VerticalRel}[0]{\ensuremath{\mathsf{V}}}
  \newcommand{\tile}[0]{\ensuremath{t}}
  \newcommand{\tileprime}[0]{\ensuremath{t^{\prime}}}

  \newcommand{\Horizontal}[0]{\ensuremath{H}}
  \newcommand{\Vertical}[0]{\ensuremath{V}}
  \newcommand{\DiagonalX}[1]{\ensuremath{D_{#1}}}
  \newcommand{\Closure}[1]{\ensuremath{E_{#1}}}

\usepackage{amssymb}
\usepackage{xspace}
\usepackage{paralist}

\usepackage{amsthm}

\theoremstyle{remark}
\newtheorem{defn}{Definition}

\input{figures.tex}

\usepackage[final]{microtype}

\makeatletter
\def\orcidID#1{\smash{\href{http://orcid.org/#1}{\protect\raisebox{-1.25pt}{
\protect\includegraphics{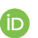}}}}}
\makeatother

\begin{document}

\title{SHACL Satisfiability and Containment\\(Extended Paper)}
\author{Paolo Pareti\inst{1}\orcidID{0000-0002-2502-0011} \and George Konstantinidis\inst{1}\orcidID{0000-0002-3962-9303} \and Fabio	Mogavero \inst{2}\orcidID{0000-0002-5140-5783} \and \\ Timothy J.\ Norman\inst{1}\orcidID{0000-0002-6387-4034} }

\authorrunning{P. Pareti et al.}
\institute{University of Southampton, Southampton, United Kingdom \\ \email{\{pp1v17,g.konstantinidis,t.j.norman\}@soton.ac.uk} \\ \and Universit\`{a} degli Studi di Napoli Federico II, Napoli, Italy \\ \email{fabio.mogavero@unina.it}}

\maketitle
\begin{abstract}
The \emph{Shapes Constraint Language (SHACL)} is a recent W3C recommendation language for validating RDF data. Specifically, SHACL documents are collections of constraints that enforce particular shapes on an RDF graph. Previous work on the topic has provided theoretical and practical results for the validation problem, but did not consider the standard decision problems of \emph{satisfiability} and \emph{containment}, which are crucial for verifying the feasibility of the constraints and important for design and optimization purposes. In this paper, we undertake a thorough study of the different features of SHACL by providing a translation to a new first-order language, called \SCL, that precisely captures the semantics of SHACL w.r.t.\ satisfiability and containment. We study the interaction of SHACL features in this logic and provide the detailed map of decidability and complexity results of the aforementioned decision problems for different SHACL sublanguages. Notably, we prove that both problems are undecidable for the full language, but we present decidable combinations of interesting features.
\end{abstract}

\input{1introduction.tex}

\input{2background.tex}

\input{3-0-foltranslation.tex}

\input{4satisfiability.tex}

\input{9conclusion.tex}

\bibliographystyle{splncs04}
\bibliography{litbib}

\appendix
\input{zappendix.tex}

\end{document}

%% file: figures.tex
\usepackage{tikz}
\usetikzlibrary{arrows.meta,shapes,calc,patterns}

\tikzstyle{every node} =
  [black, draw = none, thick, dashed, fill = none]

\tikzstyle{unk} =
  [brown]
\tikzstyle{dec} =
  [rounded rectangle, draw = blue, blue]
\tikzstyle{decthm} =
  [dec, very thick, solid]
\tikzstyle{und} =
  [rectangle, draw = red, red]
\tikzstyle{undthm} =
  [und, very thick, solid]
\tikzstyle{fmp} =
  [fill = blue!20!white]
\tikzstyle{notfmp} =
  [fill = red!10!white]

\tikzstyle{every edge} +=
  [black, very thin, dashed]

\tikzstyle{decder} =
  [blue, thick, solid]
\tikzstyle{decimp} =
  [decder, dotted]
\tikzstyle{undder} =https://www.overleaf.com/project/5e42733697cde1000189d483
  [red, thick, solid]

\renewcommand{\S}{\texttt{S}}
\renewcommand{\O}{\texttt{O}}

  \newcommand{\figfrg}
    {
    \begin{center}
    \footnotesize
    \scalebox{0.85}[0.76]
    {
    \begin{tikzpicture}
      [node distance = 4em, bend angle = 22.5]

      \node [decthm, fmp]
            (0)
            []
            {\X};

      \node [dec, fmp]
            (A)
            [above of = 0, yshift = 1em]
            {\A};
      \node [unk, notfmp]
            (O)
            [left of = A, xshift = -4em]
            {\O};
      \node [dec, fmp]
            (S)
            [left of = O, xshift = -4em]
            {\S};
      \node [dec, notfmp]
            (C)
            [right of = A, xshift = 4em]
            {\C};
      \node [dec, fmp]
            (E)
            [right of = C, xshift = 4em]
            {\E};

      \node [unk]
            (SE)
            [above of = A, yshift = 1.5em]
            {\S\,\E};
      \node [unk, notfmp]
            (AO)
            [left of = SE]
            {\A\,\O};
      \node [unk, notfmp]
            (SC)
            [left of = AO]
            {\S\,\C};
      \node [undthm, notfmp]
            (SO)
            [left of = SC]
            {\S\,\O};
      \node [dec, fmp]
            (SA)
            [left of = SO]
            {\S\,\A};
      \node [dec, notfmp]
            (AC)
            [right of = SE]
            {\A\,\C};
      \node [unk, notfmp]
            (EO)
            [right of = AC]
            {\E\,\O};
      \node [dec, notfmp]
            (EC)
            [right of = EO]
            {\E\,\C};
      \node [dec, fmp]
            (AE)
            [right of = EC]
            {\A\,\E};

      \node [unk]
            (SAE)
            [above of = SE, yshift = 1.5em]
            {\S\,\A\,\E};
      \node [und, notfmp]
            (SEO)
            [left of = SAE, xshift = -1.5em]
            {\S\,\E\,\O};
      \node [undthm, notfmp]
            (SAC)
            [left of = SEO]
            {\S\,\A\,\C};
      \node [und, notfmp]
            (SAO)
            [left of = SAC]
            {\S\,\A\,\O};
      \node [undthm, notfmp]
            (SEC)
            [right of = SAE, xshift = 1.5em]
            {\S\,\E\,\C};
      \node [unk, notfmp]
            (AEO)
            [right of = SEC]
            {\A\,\E\,\O};
      \node [dec, notfmp]
            (AEC)
            [right of = AEO]
            {\A\,\E\,\C};

      \node [undthm, notfmp]
            (SZAE)
            [above of = SAE, xshift = -2.5em]
            {\S\,\Z\,\A\,\E};
      \node [decthm, fmp]
            (SZAD)
            [above of = SA, yshift = 5.5em, xshift = -1.5em]
            {\S\,\Z\,\A\,\D};
      \node [unk]
            (SADE)
            [above of = SAE, xshift = 2.5em]
            {\S\,\A\,\D\,\E};
      \node [decthm, fmp]
            (ZADE)
            [above of = AE, yshift = 5.5em, xshift = 1.5em]
            {\Z\,\A\,\D\,\E};

      \node [und, notfmp]
            (SZADE)
            [node distance = 9em, above of = SAE]
            {\S\,\Z\,\A\,\D\,\E};
      \node [und, notfmp]
            (SAEO)
            [left of = SZADE, xshift = -2em]
            {\S\,\A\,\E\,\O};
      \node [und, notfmp]
            (SZADC)
            [left of = SAEO, xshift = -1em]
            {\S\,\Z\,\A\,\D\,\C};
      \node [und, notfmp]
            (SZADO)
            [left of = SZADC, xshift = -1em]
            {\S\,\Z\,\A\,\D\,\O};
      \node [und, notfmp]
            (SAEC)
            [right of = SZADE, xshift = 2em]
            {\S\,\A\,\E\,\C};
      \node [unk, notfmp]
            (ZADEO)
            [right of = SAEC, xshift = 1em]
            {\Z\,\A\,\D\,\E\,\O};
      \node [decthm, notfmp]
            (ZADEC)
            [right of = ZADEO, xshift = 1em]
            {\Z\,\A\,\D\,\E\,\C};

      \node [und, notfmp]
            (SZATE)
            [above of = SZADE, xshift = -2.5em]
            {\S\,\Z\,\A\,\T\,\E};
      \node [decthm, notfmp]
            (SZATD)
            [above of = SZAD, yshift = 5em]
            {\S\,\Z\,\A\,\T\,\D};
      \node [unk]
            (SATDE)
            [above of = SZADE, xshift = 2.5em]
            {\S\,\A\,\T\,\D\,\E};
      \node [unk]
            (ZATDE)
            [above of = ZADE, yshift = 5em]
            {\Z\,\A\,\T\,\D\,\E};

      \node [und, notfmp]
            (SZATDE)
            [node distance = 8em, above of = SZADE]
            {\S\,\Z\,\A\,\T\,\D\,\E};
      \node [und, notfmp]
            (SZATDO)
            [above of = SZATD]
            {\S\,\Z\,\A\,\T\,\D\,\O};
      \node [und, notfmp]
            (SZATDC)
            [node distance = 6em, right of = SZATDO]
            {\S\,\Z\,\A\,\T\,\D\,\C};
      \node [unk, notfmp]
            (ZATDEC)
            [above of = ZATDE]
            {\Z\,\A\,\T\,\D\,\E\,\C};
      \node [unk, notfmp]
            (ZATDEO)
            [node distance = 6em, left of = ZATDEC]
            {\Z\,\A\,\T\,\D\,\E\,\O};

      \node [und, notfmp]
            (SZATDEO)
            [above of = SAEO, yshift = 8em]
            {\S\,\Z\,\A\,\T\,\D\,\E\,\O};
      \node [und, notfmp]
            (SZATDOC)
            [above of = SZATDO, xshift = 3em]
            {\S\,\Z\,\A\,\T\,\D\,\O\,\C};
      \node [und, notfmp]
            (SZATDEC)
            [above of = SAEC, yshift = 8em]
            {\S\,\Z\,\A\,\T\,\D\,\E\,\C};
      \node [unk, notfmp]
            (ZATDEOC)
            [above of = ZATDEC, xshift = -3em]
            {\Z\,\A\,\T\,\D\,\E\,\O\,\C};

      \path[Latex-Latex]
        (0)     edge  [decder]
                      (A)
        (0.west)
                edge  [decder]
                      (S)
        ;
      \path[-Latex]
        (0)     edge  []
                      (O)
                edge  [decimp]
                      (C)
        (0.east)
                edge  [decimp]
                      (E)
        ;

      \path[Latex-Latex]
        (A)     edge  [decder]
                      (SA.south)
        (O.north)
                edge  []
                      (AO.south)
        (S.north)
                edge  [decder]
                      (SA.south)
        ;
      \path[-Latex]
        (A.north)
                edge  []
                      (AO.south)
                edge  [decimp]
                      (AC.south)
        (A)     edge  [decimp]
                      (AE.south)
        (O.north)
                edge  []
                      (SO.south)
                edge  []
                      (EO.south)
        (S.north)
                edge  []
                      (SE.south)
                edge  []
                      (SC.south)
                edge  []
                      (SO.south)
                edge  [decder]
                      (SA.south)
        (C.north)
                edge  []
                      (SC.south)
                edge  [decder]
                      (AC.south)
                edge  [decder]
                      (EC.south)
        (E.north)
                edge  []
                      (SE.south)
                edge  []
                      (EO.south)
                edge  [decimp]
                      (EC.south)
                edge  [decder]
                      (AE.south)
        ;

      \path[-Latex]
        (SE.north)
                edge  []
                      (SAE.south)
                edge  []
                      (SEC.south)
                edge  []
                      (SEO.south)
        (AO.north)
                edge  []
                      (SAO.south)
                edge  []
                      (AEO.south)
        (SC.north)
                edge  []
                      (SAC.south)
                edge  []
                      (SEC.south west)
        (SO.north)
                edge  [undder]
                      (SAO.south)
                edge  [undder]
                      (SEO.south)
        (SA.north)
                edge  []
                      (SAE.south)
                edge  []
                      (SAC.south)
                edge  []
                      (SAO.south)
        (AC.north)
                edge  []
                      (SAC.south)
                edge  [decder]
                      (AEC.south)
        (EO.north)
                edge  []
                      (SEO.south east)
                edge  []
                      (AEO.south)
        (EC.north)
                edge  []
                      (SEC.south)
                edge  [decder]
                      (AEC.south)
        (AE.north)
                edge  []
                      (SAE.south)
                edge  [decimp]
                      (AEC.south)
                edge  []
                      (AEO.south)
        ;

      \path[-Latex]
        (SA.north)
                edge  [decimp]
                      (SZAD.south)
        (AE.north)
                edge  [decder]
                      (ZADE.south)
        ;

      \path[-Latex]
        (SAE)   edge  []
                      (SZAE)
                edge  []
                      (SADE)
        ;

      \path[-Latex]
        (SAE)   edge  [bend angle = 40, bend left]
                      (SAEO.south east)
                edge  [bend angle = 40, bend right]
                      (SAEC.south west)
        (SEO)   edge  [undder]
                      (SAEO.south)
        (SAC.north)
                edge  [undder]
                      (SZADC.south)
                edge  [undder, bend left]
                      (SAEC.south west)
        (SAO.north)
                edge  [undder]
                      (SZADO.south)
                edge  [undder]
                      (SAEO.south)
        (SEC)   edge  [undder]
                      (SAEC.south)
        (AEO.north)
                edge  []
                      (ZADEO.south)
                edge  [bend right]
                      (SAEO.south east)
        (AEC.north)
                edge  [decder]
                      (ZADEC.south)
                edge  []
                      (SAEC.south)
        ;

      \path[-Latex]
        (SZAE.north)
                edge  [undder]
                      (SZADE.south)
        (SZAD)  edge  []
                      (SZADE.south)
        (SZAD.north)
                edge  []
                      (SZADO.south)
                edge  []
                      (SZADC.south)
        (SADE.north)
                edge  []
                      (SZADE.south)
        (ZADE)  edge  []
                      (SZADE.south)
        (ZADE.north)
                edge  []
                      (ZADEO.south)
                edge  [decimp]
                      (ZADEC.south)
        ;

      \path[-Latex]
        (SZAE.north)
                edge  [undder]
                      (SZATE.south)
        (SZAD.north)
                edge  [decimp, bend left]
                      (SZATD.south)
        (SADE.north)
                edge  []
                      (SATDE.south)
        (ZADE.north)
                edge  [bend right]
                      (ZATDE.south)
        ;

      \path[-Latex]
        (SZADE) edge  [undder]
                      (SZATDE)
        (SZADC) edge  [undder]
                      (SZATDC.south)
        (SZADO) edge  [undder, bend angle = 35, bend right]
                      (SZATDO)
        (ZADEO) edge  []
                      (ZATDEO.south)
        (ZADEC) edge  [bend angle = 35, bend left]
                      (ZATDEC)
        ;

      \path[-Latex]
        (SAEO)  edge  [undder]
                      (SZATDEO.south)
        (SAEC)  edge  [undder]
                      (SZATDEC.south)
        ;

      \path[-Latex]
        (SZATE) edge  [undder]
                      (SZATDE)
        (SZATD) edge  []
                      (SZATDE)
        (SZATD.north)
                edge  []
                      (SZATDO.south)
                edge  []
                      (SZATDC.south)
        (SATDE) edge  []
                      (SZATDE)
        (ZATDE) edge  []
                      (SZATDE)
        (ZATDE.north)
                edge  []
                      (ZATDEO.south)
                edge  []
                      (ZATDEC.south)
        ;

      \path[-Latex]
        (SZATDE.north)
                edge  [undder]
                      (SZATDEO.south)
                edge  [undder]
                      (SZATDEC.south)
        (SZATDC.north)
                edge  [undder]
                      (SZATDEC.south west)
                edge  [undder]
                      (SZATDOC.south)
        (SZATDO.north)
                edge  [undder]
                      (SZATDEO.south)
                edge  [undder]
                      (SZATDOC.south)
        (ZATDEO.north)
                edge  []
                      (SZATDEO.south east)
                edge  []
                      (ZATDEOC.south)
        (ZATDEC.north)
                edge  []
                      (SZATDEC.south)
                edge  []
                      (ZATDEOC.south)
        ;

    \end{tikzpicture}
    }
    \end{center}
    }

\newcommand{\ie}{\emph{i.e.}\xspace}
\newcommand{\wrt}{\emph{w.r.t.}\xspace}

%% file: 1introduction.tex
\section{Introduction}
The Shapes Constraint Language (SHACL) has been recently introduced as a W3C recommendation language for the validation of RDF graphs. A SHACL document is a collection of \emph{shapes} which define particular constraints and specify which nodes in a graph should be validated against these constraints. The ability to validate data with respect to a set of constraints is of particular importance for RDF graphs, as they are schemaless by design. Validation can be used to detect problems in a dataset and it can provide data quality guarantees for the purpose of data exchange and interoperability.

Recent work has focused on defining precise semantics and implementations for validation of SHACL documents, in particular for the case of recursion \cite{Corman2018SHACL,SHACLstableModelSemantics}. 
In this paper, instead, we focus on the decision problems of satisfiability and containment for SHACL documents; problems which have not been previously investigated. Given a particular SHACL document, satisfiability is the problem of deciding whether there is an RDF graph which is validated by the document; we also investigate  finite satisfiability, that is, whether there exists a valid graph of finite size. 
Containment studies whether a particular SHACL document is subsumed by a second one; that is, whether all graphs that are validated by the first are also validated by the second. 
We investigate whether these decision problems can be decided not only at the level of documents, but also for individual shapes (i.e. sets of constraints) within documents.

Satisfiability and containment are standard decision problems that have important applications in optimization and design. When integrating two datasets subject to two different SHACL documents, for example, it is important to know whether the two SHACL documents are in conflict with each other, or if one of them is subsumed by the other. At the level of shapes, an unsatisfiable shape constraint might not necessarily cause the unsatisfiability of a whole SHACL document, but it is likely an indication of a design error. Being able to decide containment for individual shapes offers more design choices to the author of a SHACL document, and it is a venue for optimization.

In this paper we focus on the \emph{core constraint components} of SHACL \cite{2017SHACL} and we do not consider recursion. Validation under recursion is left unspecified in SHACL and, while different semantics have been proposed \cite{Corman2018SHACL,SHACLstableModelSemantics}, we already show that even without it the language has undecidable satisfiability and containment. For a subset of the core constraint components
and a restricted form of recursion (\`a la stratified negation), containment of individual shape constraints is shown to be decidable in  \cite{martin2020shapecontainment}. This is achieved via reduction to description logic reasoning~\cite{BCMNP03}, reminiscent to our Thm~\ref{thm:sat(sza)}.

One of our main contributions is a comprehensive translation of SHACL into \SCL , a new fragment of first-order logic extended with counting quantifiers and the transitive closure operator. To the best of our knowledge such a translation has not been attempted before. Previous formalisations of SHACL semantics were either limited to certain aspects of the specification, such as constraints \cite{Corman2018SHACL}, or to subsets of the core constraint components \cite{pareti2019c}.
Our approach, instead, translates a SHACL document to an \SCL\ equisatisfiable sentence, i.e., there is a valid RDF graph for the first if and only if there is a model for the second. 

Distinct SHACL constructs translate to particular
\SCL\ features of different expressiveness. 
We identify eight such prominent features (such as counting quantifiers or transitive closure) that can be used on top of a base logic and study their interactions. 
On one hand, the full language is undecidable and, in fact, so are most fragments with just three or four features. On the other hand, our base language has decidable satisfiability and containment, and it is ExpTime-complete. We create a detailed map, in between these extremes, proving positive and negative results for many interesting combinations.

%% file: 2background.tex
\section{Background and Problem Definition}
\label{sec:problem}

The core structure of the RDF data model is a graph whose nodes and edges are defined by a set of \emph{triples}. A triple \trip{s}{p}{o} identifies an edge with label $p$, called \emph{predicate}, from a node $s$, called \emph{subject}, to a node $o$, called \emph{object}. The main type of entities that act as nodes and edges in RDF graphs are IRIs. We represent RDF graphs in Turtle syntax and by abbreviating IRIs using XML namespaces; the namespace \texttt{sh} refers to SHACL terms.

In an RDF graph, \emph{literal} constants (representing datatype values) can only appear in the object position of a triple, while in \emph{generalized} RDF \cite{Hayes2014GeneralisedRDF} they can appear in any position. We will use the generalised model for simplicity. Most of our results apply to both data models and we will state clearly when this is not the case. 
We do not use variables in the predicate position in this paper and so we represent triples as binary relations in FOL. We use the atom $R(s,o)$ as a shorthand for $\trip{s}{R}{o}$. We use a minus sign to identify the \emph{inverse} atom, namely $R^{-}(s,o) = R(o,s)$. We use the binary relation name \isA to represent class membership triples $\trip{s}{\texttt{rdf:type}}{o}$ as $\isA(s,o)$.

\begin{figure}[t]
\begin{minipage}[t]{0.59\linewidth}
\begin{lstlisting}
:studentShape a sh:NodeShape ;
    sh:targetClass :Student ;
    sh:not :disjFacultyShape .
:disjFacultyShape a sh:PropertyShape ;
      sh:path (:hasSupervisor :hasFaculty);
      sh:disjoint :hasFaculty .
\end{lstlisting}
\end{minipage}\hfill
\begin{minipage}[t]{0.39\linewidth}
\begin{lstlisting}
:Alex a :Student ;
  :hasFaculty :CS ;
  :hasSupervisor :Jane .
:Jane :hasFaculty :CS .
\end{lstlisting}
\end{minipage}
\vspace{-1em}
\captionof{figure}{A SHACL document (left) and a graph that validates it (right).}
\label{fig:MainExample}
\end{figure}

SHACL defines constraints that can validate \texttt{RDF} graphs \cite{2017SHACL}.
A SHACL document is a set of \emph{shapes}. A shape, denoted \vshape\pair{\vshapet}{\vshapec}, has three main components: (1) a set of \emph{constraints} which are used in conjunction, and hence referred to as a single constraint $\vshapec$; (2) a set of \emph{target declarations}, referred to as \emph{target definition} \vshapet , which provides a set of RDF nodes that are validated against $\vshapec$; and (3) a \emph{shape name} $\vshape$. One can think of $\vshapet$ and $\vshapec$ as unary queries over the nodes of $G$.
Given a node $\vn$ in a graph $G$, and a shape \vshape\pair{\vshapet}{\vshapec}, we denote with $G \models \vshapet(\vn)$ 
the fact that node $\vn$ that satisfies definition $\vshapet$, and $G \models \vshapec(\vn)$ denotes that a node $\vn$ validates $\vshapec$ in $G$.
A graph $G$ validates a shape \vshape\pair{\vshapet}{\vshapec}, formally $G \models \vshape\pair{\vshapet}{\vshapec}$, iff every node in the target $\vshapet$ validates the constraints \vshapec, that is, 
iff for all $\vn \in G,$ if $G \models \vshapet(\vn)$ then $G \models \vshapec(\vn)$.
An empty target definition is never satisfied while an empty constraint definition is always satisfied. A graph $G$ validates a set of shape definitions, i.e.\ a SHACL document, $\vShapeDocument$, formally $G \models \vShapeDocument$, iff $G$ validates all the shapes in $\vShapeDocument$.
Constraints might refer to other shapes.
When a shape is referenced by another shape it can be handed down a set of  \emph{focus nodes} to validate, in addition to those from its own target definition. A shape is \emph{recursive} when it references itself (directly or through other shapes). As mentioned, we focus on non-recursive SHACL documents using the SHACL core constraint components. Without loss of generality, we assume that shape names in a SHACL document do not occur in other SHACL documents or graphs. 

The example SHACL document in Figure \ref{fig:MainExample} defines the constraint that, intuitively, all students must have at least one supervisor from the same faculty. The shape with name \texttt{:studentShape} has class \texttt{:Student} as a target, meaning that all members of this class must satisfy the constraint of the shape. The constraint definition of \texttt{:studentShape} requires the non-satisfaction of shape \texttt{:disjFacultyShape}, i.e., a node satisfies \texttt{:studentShape} if it does not satisfy \texttt{:disjFacultyShape}. The \texttt{:disjFacultyShape} shape states that an entity has no faculty in common with any of their supervisors (the \sh{path} term defines a property chain, i.e., a composition of roles \texttt{:hasSupervisor} and \texttt{:hasFaculty}). A graph that validates these shapes is provided in Figure \ref{fig:MainExample}. It can be made invalid by changing the faculty of \texttt{:Jane} in the last triple.

We now define the SHACL satisfiability and containment problems.

\begin{enumerate}[(i)]
     \item
        \textbf{SHACL Satisfiability}: A SHACL document $\vShapeDocument$ is satisfiable iff there exists a graph $G$ such that $G \models \vShapeDocument$.
    \item
        \textbf{Constraint Satisfiability}: A SHACL constraint \vshapec\ 
        is satisfiable iff there exists a graph $G$ and a node $\vn$ such that $G \models \vshapec(\vn)$.
    \item
        \textbf{SHACL Containment}: For all SHACL documents $\vShapeDocument_1$, $\vShapeDocument_2$, we say that $\vShapeDocument_1$ is contained in $\vShapeDocument_2$, denoted $\vShapeDocument_1 \subseteq$ $\vShapeDocument_2$, iff for all graphs $G$, if $ G \models \vShapeDocument_1$ then $G \models \vShapeDocument_2$.
    \item
        \textbf{Constraint Containment}: For all SHACL constraints $\vshapec_1$ and $\vshapec_2$ 
        we say that $\vshapec_1$ is contained in $\vshapec_2$, 
        denoted by  $\vshapec_1, \subseteq \vshapec_2$ iff for all graphs $G$ and nodes \vn , 
 if $G \models \vshapec_1(\vn)$ then $G \models \vshapec_2(\vn)$.
\end{enumerate}

The satisfiability and containment problems for constraints can be reduced to SHACL satisfiability, as follows.
A constraint \vshapec\ is satisfiable iff there exists a constant \conc , either occurring in \vshapec\ or a fresh one, such that the SHACL document corresponding to shape $\vshape\pair{\vshapet_{\conc}}{\vshapec}$ is satisfiable, where $\vshapet_{\conc}$ is the target definition that targets node $\conc$. Similarly, constraint $\vshapec_1$ is not contained in $\vshapec_2$ iff there exists a constant \conc , occurring in $\vshapec_1$, $\vshapec_2$ or a fresh one, such that the SHACL document corresponding to shape $\vshape\pair{\vshapet_{\conc}}{\vshapec'}$ is satisfiable; $\vshapec'(\vx)$ is true whenever $\vshapec_1(\vx)$ is true and $\vshapec_2(\vx)$ if false. Thus, satisfiability and containment of constraints in a given SHACL fragment are decidable whenever SHACL satisfiability of that fragment is decidable, and have the same complexity upper bound. However, undecidability of SHACL satisfiability in a fragment does not necessarily imply undecidability for the two constraint problems; we leave this as an open problem.

%% file: 3-0-foltranslation.tex
\section{A First Order Language for SHACL Documents} \label{sec:language}
In this section we present a translation of SHACL into an equisatisfiable fragment of FOL extended with counting quantifiers and the transitive closure operator, called \SCL. As discussed before, for a shape \vshape\pair{\vshapet}{\vshapec}\ in a SHACL document \vShapeDocument , $t$ and $d$ can be seen as unary queries. 
Intuitively, given a suitable translation $q$ from SHACL into FOL, \vShapeDocument\ is satisfiable iff the sentence
$\bigwedge_{\vshape\pair{\vshapet}{\vshapec} \in \vShapeDocument} \forall x . \; q(\vshapet(\vx)) \rightarrow q(\vshapec(\vx))$ 
is satisfiable, i.e., a node in the target definition of a shape needs to satisfy its constraint, for every shape. We subsequently present an approach that constructs such a sentence. This is reminiscent of  \cite{SHACL2SPARQLtranslation}, where a SHACL document \vShapeDocument\ is translated into a SPARQL query that is true on graphs which however violate \vShapeDocument.
Intuitively, this query corresponds to sentence $\bigvee_{\vshape\pair{\vshapet}{\vshapec} \in \vShapeDocument} \exists x . q(\vshapet(\vx)) \wedge \neg q(\vshapec(\vx))$, i.e.\ the negation of the sentence above. Nevertheless, several assumptions made in \cite{SHACL2SPARQLtranslation}, such that ordering two values is not more complex than checking their equivalence, do not hold for the purposes of satisfiability and containment. Therefore the translation from \cite{SHACL2SPARQLtranslation} cannot be used to reduce SHACL satisfiability and containment to the satisfiability and containment of SPARQL queries. We will use $\tau$ to denote the translation function from a SHACL document $\vShapeDocument$ to an \SCL\ sentence $\tau(\vShapeDocument)$, which is polynomial in the size of $\vShapeDocument$ and computable in polynomial time.
We refer to our appendix\footnote{\url{http://w3id.org/asset/ISWC2020}} for the complete translations of $\tau$ and its inverse $\tau^{-}$.

Next, we present our grammar of \SCL\ in Def.~\ref{def:syn}.  
For simplicity, we assume that target definitions contain at most one target declaration, and that shapes referenced by other shapes have an empty target definition.
This does not affect generality, as any shape can be trivially split in multiple copies: one per target declaration and one without any. Letters in square brackets in Def.~\ref{def:syn} are annotations naming \SCL\  features and thus are not part of the grammar. 
The top-level symbol $\varphi$ in \SCL\ corresponds to a SHACL document. This could be empty ($\top$), a conjunction of documents, or the translation of an individual shape. A sentence that corresponds to a single shape could have five different forms in \SCL, depending on the target definition of the translated shape.
These are summarized in Table~\ref{tab:translationTargs}, where $\taus{\vshapec}{\vx}$ is the \SCL\ translation of the constraint of the shape. 
In SHACL only four types of target declarations are allowed:  (1) a particular constant \texttt{c} (node target), (2) instances of class \texttt{c} (class target), or (3)/(4) subjects/objects of a triple with predicate \texttt{R} (subject-of/object-of target). Our translation function gives explicit names to referenced shapes using the \hasshapePredicate relation. 
We refer to the last component of the $\varphi$ rule (i.e., $\forall \vx . \; \hasshape{\vx}{\sconst} \leftrightarrow \psi(\vx)$) as a \emph{\namedshapedefinition } and to its internal constant $\sconst$ as \emph{\namedshape }.

\begin{table}[t]
\begin{center}
\caption{ 
Translation of shape $\vshape\pair{\vshapet}{\vshapec}$ in SCL with respect to its target definition $\vshapet$. 
} 
\label{tab:translationTargs}
\begin{tabular}{ |l | l |}
 \hline
 Target declaration in $\vshapet$  & Translation $\tau(\vshape\pair{\vshapet}{\vshapec})$ \\ \hline 
 Node target (node \conc) & $\taus{\vshapec}{\conc}$ (equivalent form of: $\forall \vx . \; \vx =\conc \; \rightarrow \taus{\vshapec}{\vx} $ ) \\ \hline
 Class target (class \texttt{c})  & $\forall \vx . \isA(\vx,$\texttt{c}$)\rightarrow \taus{\vshapec}{\vx}$\\ \hline
 Subjects-of target (relation $R$)  & $\forall \vx, \vy . R(\vx,\vy) \rightarrow \taus{\vshapec}{\vx}$\\ \hline
 Objects-of target (relation $R$)  & $\forall \vx, \vy . R^{-}(\vx,\vy) \rightarrow \taus{\vshapec}{\vx}$ \\ \hline
 No target declaration & $\forall \vx . \; \hasshape{\vx}{\sconst} \leftrightarrow \taus{\vshapec}{\vx}$ \\ \hline
\end{tabular}
\end{center}
\end{table}

The non terminal symbol $\psi(\vx)$ corresponds to the subgrammar of the SHACL constraints. 
Within this subgrammar, $\top$ identifies an empty constraint, $\vx = \conc$ a constant equivalence constraint and $F$ a monadic filter relation (e.g.\ $F^{\isIRI}(\vx)$, true iff $\vx$ is an IRI).
By \emph{filters} we refer to the SHACL constraints about ordering, node-type, datatype, language tag, regular expressions and string length. Filters are captured by $F(\vx)$ and the \lO\ component. The \C\ component captures qualified value shape cardinality constraints. The \E, \D\ and \lO\ components capture the equality, disjointedness  and order property pair components. 
The $\pi(\vx, \vy)$ subgrammar models SHACL property paths. Within this subgrammar \S\ denotes sequence paths, \A\ denotes alternate paths, \Z\ denotes a zero-or-one path and \T\ denotes a zero-or-more path.

\begin{defn}
    \label{def:syn}
    The \emph{SHACL} first-order language (\SCL, for short) is the set of \emph{sentences ($\varphi$)} and \emph{one-variable formulas ($\psi(\vx)$)} built according to the following context-free grammar, where $\conc$ and $\sconst$ are constants (from disjoint domains), $F$ is a monadic-filter name,
    $R$ is a binary-relation name, $^{\star}$ indicates the transitive closure of the relation induced by $\pi(\vx,\vy)$, the superscript $\pm$ refers to a relation or its inverse, and $n \in \mathbb{N}$:
    \begin{align*}
      \varphi \gramdef\;
        & \top \mid \psi(\conc) \mid \forall \vx \,.\, \isA(\vx, \conc) \rightarrow
          \psi(\vx) \mid \forall \vx, \vy \,.\, R^{\pm}(\vx, \vy) \rightarrow
          \psi(\vx) \mid \varphi \wedge \varphi; \mid \\
          & \forall \vx . \; \hasshape{\vx}{\sconst} \leftrightarrow \psi(\vx) \, ; \\
      \psi(\vx) \gramdef\;
        & \top \mid \vx = \conc \mid F(\vx) \mid \hasshape{\vx}{\sconst} \,  \mid \neg \psi(\vx) \mid \psi(\vx)
          \wedge \psi(\vx) \mid  \\
        & \exists \vy .\, \pi(\vx, \vy) \wedge \psi(\vy)
          \mid \neg \exists \vy .\, \pi(\vx, \vy) \wedge R(\vx, \vy)
          \,\text{\textbf{[\D]}} \mid \forall \vy .\, \pi(\vx, \vy)
          \leftrightarrow R(\vx, \vy) \,\text{\textbf{[\E]}} \mid \\
        & \forall \vy, \vz \,.\, \pi(\vx, \vy) \wedge R(\vx, \vz) \rightarrow
          \sigma(\vy, \vz) \,\text{\textbf{[\O]}} \mid \exists^{\geq n} \vy
          \,.\, \pi(\vx, \vy) \wedge \psi(\vy) \,\text{\textbf{[\C]}} ; \\
      \pi(\vx, \vy) \gramdef\;
        & R^{\pm}(\vx, \vy) \mid \exists \vz \,.\, \pi(\vx, \vz) \!\wedge\!
          \pi(\vz, \vy) \,\text{\textbf{[\S]}} \mid \vx \!=\! \vy \!\vee\!
          \pi(\vx, \vy) \,\text{\textbf{[\Z]}} \mid \pi(\vx, \vy) \!\vee\!
          \pi(\vx, \vy) \,\text{\textbf{[\A]}} \mid \\
        & (\pi(\vx, \vy))^{\star} \,\text{\textbf{[\T]}}; \\
      \sigma(\vx, \vy) \gramdef\;
        & \vx <^{\pm} \vy \mid \vx \leq^{\pm} \vy.  
    \end{align*}
  \end{defn}
To enhance readability, we define the following syntactic shortcuts: 
\begin{enumerate}[(i)]
    \item
        $\psi_{1}(\vx) \vee \psi_{2}(\vx) \grameq \neg (\neg \psi_{1}(\vx) \wedge \neg \psi_{2}(\vx))$;
    \item
        $\pi(\vx, \conc) \grameq \exists \vy . \pi(\vx, \vy) \wedge \vy = \conc$;
    \item
        $\forall \vy \,.\, \pi(\vx, \vy) \rightarrow \psi(y) \grameq \neg \exists \vy \,.\, \pi(\vx, \vy) \wedge \neg \psi(\vy)$.
\end{enumerate}

\begin{figure}[t]
\begin{minipage}[t]{0.48\linewidth}
\begin{lstlisting}

select ?x where {
    ?x rdf:type :Student .
    filter not exists {
        ?x :hasSupervisor ?z .
        ?z :hasFaculty ?y .
        ?x :hasFaculty ?y .  } }
\end{lstlisting}
\end{minipage}\hfill
\begin{minipage}[t]{0.50\linewidth}
\begin{align*}
( \forall \vx . \; & \isA(\vx, \texttt{:Student}) \rightarrow  \\
& \neg    \hasshape{\vx}{\texttt{:disjFacultyShape}} ) \:\wedge \\
( \forall \vx . &\; \hasshape{\vx}{ \texttt{:disjFacultyShape}} \leftrightarrow  \\ &
\neg \exists \vy . \; ( \; \exists \vz . \;   R_{\texttt{:hasSupervisor}}(\vx, \vz) \; \wedge
\\ & \hphantom{\neg \exists \vy . ( \; \exists \vz . \;\;} R_{\texttt{:hasFaculty}}(\vz, \vy) \; \wedge
\\ & \hphantom{\neg \exists \vy . ( \; \exists \vz . \;\;} R_{\texttt{:hasFaculty}}(\vx, \vy) \; ) \; )
\end{align*}
\end{minipage} 
\captionof{figure}{Translation of the SHACL document  from Fig. \ref{fig:MainExample} into the SPARQL query that looks for violations (left) and into an \SCL\ sentence (right). } 
\label{fig:TranslationExample}
\end{figure}

Our translation $\tau$ results in a subset of \SCL\ sentences, called \emph{well-formed}. An \SCL\ sentence is well-formed if for every occurrence of a \namedshape\ $s$ there is a corresponding \namedshapedefinition\ sentence with the same $s$, and no \namedshapedefinition s are recursively defined.  
Fig. \ref{fig:TranslationExample} shows the translation of the document from Fig. \ref{fig:MainExample}, into a SPARQL query, via \cite{SHACL2SPARQLtranslation},  
and a well-formed \SCL\ sentence, via $\tau$.  

To distinguish different fragments of \SCL , Table \ref{tab:SHACLcomponentsInOurGrammar} lists a number of \emph{prominent} SHACL components, that is, important for the purpose of satisfiability. The language defined without any of these constructs is our \emph{base} language, denoted \X .
When using such an abbreviation of a prominent feature, we refer to the fragment of our logic that includes the base language together with that feature enabled. 
For example, \S \A\ identifies the fragment that only allows the base language, sequence paths and alternate paths.

The SHACL specification presents an unusual asymmetry in the fact that equality, disjointedness and order components forces one of their two path expressions to be an atomic relation. 
This can result in situations where the order constraints can be defined in just one direction, since only the less-than and less-than-or-equal property pair constraints are defined in SHACL. 
Our \O\ fragment models a more natural order comparison that includes the $>$ and $\geq$ components. We instead denote with \O' the fragment where the order relations in the $\sigma(\vx, \vy)$ subgrammar cannot be inverted. 

\begin{table}[t]
\begin{center}
\caption{Relation between prominent SHACL components and \SCL\ expressions.}\label{tab:SHACLcomponentsInOurGrammar}
\begin{tabular}{ | l | l |l | l |}
 \hline
 Abbr. & Name  & SHACL component & Corresponding expression  \\ \hline
 \S & Sequence Paths & Sequence Paths & $\exists \vz \,.\, \pi(\vx,\vz) \wedge \pi(\vz,\vy)$ \\ \hline
 \Z & Zero-or-one  Paths  & \sh{zeroOrOnePath} & $\vx = \vy \vee \pi(\vx, \vy)$ \\ \hline
  \A & Alternative Paths & \sh{alternativePath} & $\pi(\vx, \vy) \vee \pi(\vx, \vy)$ \\ \hline
  \T & Transitive Paths  & \begin{tabular}{@{}l@{}}\sh{zeroOrMorePath} \\ \sh{oneOrMorePath}
 \end{tabular}  & $(\pi(\vx, \vy))^{\star}$ \\ \hline
 \D & Property Pair Disjointness & \sh{disjoint} & $\neg \exists \vy . \pi(\vx, \vy) \wedge R(\vx, \vy)$ \\ \hline
 \E & Property Pair Equality & \sh{equals} & $\forall \vy \,.\, \pi(\vx,\vy) \leftrightarrow R(\vx,\vy)$ \\ \hline
 \O & Property Pair Order  & \sh{lessThanOrEquals} & \begin{tabular}{@{}l@{}}$\vx \leq^{\pm} \vy \text{ and } \vx <^{\pm} \vy $ 
 \end{tabular} \\ \hline
 \C & Cardinality Constraints  & \begin{tabular}{@{}l@{}}\sh{qualifiedValueShape} \\ \sh{qualifiedMinCount} \\ \sh{qualifiedMaxCount}\end{tabular}   & \begin{tabular}{@{}l@{}}$\exists^{\geq n} \vy \,.\, \pi(\vx, \vy) \wedge \psi(\vy)$ \\ with $n \not = 1$\end{tabular}\\ \hline
\end{tabular} 
\end{center}
\end{table} 

Relying on the standard FOL semantics, we define the satisfiability and containment for \SCL\ sentences, as well as the closely related finite-model property, in the natural way.

\begin{description}

\item[\SCL\ Sentence Satisfiability] An \SCL\ sentence $\phi$ is satisfiable iff there exists a first-order structure $\Omega$ such that $\Omega \models \phi$. 

\item[\SCL\ Sentence Containment]
For all \SCL\ sentences $\phi_1$, $\phi_2$, we say that $\phi_1$ is contained in $\phi_1$, denoted $\phi_1 \subseteq$ $\phi_2$, iff,
for all first-order structures $\Omega$, if $\Omega \models \phi_1$ then $\Omega \models \phi_2$. 

\item[\SCL\ Finite-model Property ]
An \SCL\ sentence $\phi$ (resp. \ formula $\psi(\vx)$) enjoys the finite-model property iff whenever $\phi$ is satisfiable, it is so on a finite model.
\end{description} 

In the following two subsections, we discuss SHACL-to-\SCL\ satisfiability and containment. In this respect, we assume that filters are interpreted relations. In particular, we prove equisatisfiability of SHACL and \SCL\ on models that we call \emph{canonical}, that is, having the following properties: (1) the domain of the model is the set of RDF terms, (2) such a model contains built-in interpreted relations for filters, and (3) ordering relations $<^{\pm}$ and $\leq^{\pm}$ are the disjoint union of the total orders of the different comparison types allowed in SPARQL. 
In Sec. \ref{sec:filters}, we discuss an explicit axiomatization of the semantics of a particular set of filters in order to prove decidability of the satisfiability and containment problems for several \SCL\ fragments in the face of these filters. 

\subsection{SHACL Satisfiability}\label{sec:sat}

A fine-grained analysis of the bidirectional translation between our grammar and SHACL, provided in the appendix, can lead to an inductive proof of equisatisfiability between the two languages. In particular, given a satisfiable SHACL document $\vShapeDocument$ which validates an RDF graph $G$, we can translate $G$ and $\vShapeDocument$ into a canonical first-order structure $\structureOmegaMain$ which models $\tau({\vShapeDocument})$, thus proving the latter satisfiable, and vice versa. Intuitively, the structure $\structureOmegaMain$ is composed of two substructures, \structureOmegaG\ which corresponds to the translation of triples from $G$, and \structureOmegaS\ which interprets the \hasshapePredicate\ relation. These substructures, as explained below, have disjoint interpretations and we write $\structureOmegaMain = \structureOmegaG \cup \structureOmegaS$ to denote that $I$ is the structure that considers the union of their domains and of their interpretations. 

For any RDF predicate $R$ in $G$, the structure $\structureOmegaG$ is a canonical structure that interprets the binary relation $R$ as the set of all pairs $\langle s,o \rangle$ for which $\trip{s}{R}{o}$ is in $G$. 
The structure \structureOmegaS\ interprets \hasshapePredicate\ as the binary relation which, for all referenced shape definitions $\forall \vx . \; \hasshape{\vx}{\sconst} \leftrightarrow \psi(\vx)$ in $\tau({\vShapeDocument})$, it contains a pair $\langle c,s \rangle$ whenever \structureOmegaG\ satisfies $\psi(c)$.
We will call $\structureOmegaS$ the \textit{shape definition model} of $G$ and $\vShapeDocument$. Since we do not address recursive shape definitions, this model always exists (corresponding to the \emph{faithful total assignment} from \cite{Corman2018SHACL}).  
Inversely, given a well-formed $\SCL$ sentence $\phi$ that is satisfiable and has a model $I$, by eliminating from $I$ all references of \hasshapePredicate and then transforming the elements of the relations to triples we get an RDF graph $G$ that is valid w.r.t.\ the SHACL document $\tau^{-}(\phi)$.  

\begin{theorem} \label{theorem1}
    For all SHACL documents $\vShapeDocument$:
    \begin{inparaenum}[(1)]
        \item
            $\tau(\vShapeDocument)$ is polynomially computable;
        \item
            $\vShapeDocument$ is (finitely) satisfiable iff $\tau(\vShapeDocument)$ is (finitely) satisfiable on a canonical model.
    \end{inparaenum}
      \\  For all  well-formed \SCL\ sentences $\phi$:
    \begin{inparaenum}[(1)]
        \item
            $\tau^{-}(\phi)$ is polynomially computable;
        \item
            $\phi$ is (finitely) satisfiable on canonical models iff $\tau^{-}(\phi)$ is (finitely) satisfiable.
    \end{inparaenum}
\end{theorem}

\subsection{SHACL Containment}

Containment of two SHACL documents does not immediately correspond to the containment of their SCL translations. Given two SHACL documents $\vShapeDocument_1$ and $\vShapeDocument_2$ where $\vShapeDocument_1$  is contained in $\vShapeDocument_2$, there might exist a first-order structure $\structureOmegaMain$  
that models $\tau(\vShapeDocument_1)$ but not $\tau(\vShapeDocument_2)$. Notice, in fact, that structure $\structureOmegaMain = \structureOmegaG \cup \structureOmegaOne$ models $\vShapeDocument_1$, but that $\structureOmegaOne$ does not necessarily model the referenced shape definitions of $\tau(\vShapeDocument_2)$. 
Let $\delta(\phi)$ be the definitions of \namedshape s in an \SCL\ sentence $\phi$.  
Note that for a graph $G$ and a SHACL document \vShapeDocument\ the shape definition model \structureOmegaS\ models $\delta(\tau(\vShapeDocument))$. The reduction of SHACL containment into \SCL\ is, therefore, as follows. This result also applies for containment over finite structures.

\begin{theorem} \label{theorem2}
    For all SHACL documents $\vShapeDocument_1$ and $\vShapeDocument_2$:
    \begin{inparaenum}[(1)]
        \item
            $\delta(\tau(\vShapeDocument_{2}))$ is polynomially computable;
        \item
            $\vShapeDocument_1 \subseteq \vShapeDocument_2$ iff $\tau(\vShapeDocument_1) \wedge \delta(\tau(\vShapeDocument_2)) \subseteq \tau(\vShapeDocument_2)$ on all canonical models.
    \end{inparaenum}
\end{theorem} 
\begin{proof} 

($\Rightarrow$)\ Let $\vShapeDocument_1 \subseteq \vShapeDocument_2$. If $\vShapeDocument_1$ is not satisfiable the theorem holds. If $\vShapeDocument_1$ is satisfiable, let $G$ be any graph that validates $\vShapeDocument_1$, and thus $\vShapeDocument_2$. It holds that $\structureOmegaG \cup \structureOmegaOne$ models $\tau(\vShapeDocument_1)$ per Sec.~\ref{sec:sat}, and $\structureOmegaG \cup \structureOmegaTwo$ models $\tau(\vShapeDocument_2)$. 
It is easy to see that if $\structureOmegaG \cup \structureOmegaOne$ models $\tau(\vShapeDocument_1)$ the union of another \hasshapePredicate interpretation over a disjoint set of shape names, i.e., $\structureOmegaG \cup \structureOmegaOne \cup \structureOmegaTwo$ also models $\tau(\vShapeDocument_1)$. Similarly $\structureOmegaG \cup \structureOmegaOne \cup \structureOmegaTwo$ models $\tau(\vShapeDocument_2)$ as well. 

($\Leftarrow$) If $\vShapeDocument_1$ is not contained in $\vShapeDocument_2$, then there is a graph $G$ that models $\vShapeDocument_1$ but not $\vShapeDocument_2$.
Thus, $\structureOmegaG \cup \structureOmegaOne$ models $\tau(\vShapeDocument_1)$ but $\structureOmegaG \cup \structureOmegaTwo$ does not model $\tau(\vShapeDocument_2)$. So we have that $\structureOmegaG \cup \structureOmegaOne \cup \structureOmegaTwo$ models $\tau(\vShapeDocument_1) \cup \delta(\tau(\vShapeDocument_2))$  but not $\tau(\vShapeDocument_2)$.
\end{proof}    

Since our grammar is not closed under negation we cannot trivially reduce (finite) \SCL\  containment to (finite) \SCL\ satisfiability. Nevertheless, all positive (decidability and complexity) results are obtained by exhibiting inclusion of some \SCL\ fragment into a particular (extension of a) fragment of first-order logic already studied in the literature that is closed under negation. Thus we can always solve the (finite) \SCL\ containment problem for sentences $\phi_1 \subseteq \phi_2$ by deciding (finite) unsatisfiability of a sentence $\phi_1 \wedge \neg \phi_2$.  
Dually, the unsatisfiability of an \SCL\ sentence $\phi$ is equivalent to $\phi \subseteq \bot$. Hence, containment and unsatisfiability have the same complexity. 

\input{3-3-filters}

%% file: 3-3-filters.tex
\subsection{Filter Axiomatization} \label{sec:filters}

Decidability of \SCL\ satisfiability depends on the decidability of filters. In this section we present a decidable axiomatization that allows us to treat some filters as simple relations instead of interpreted ones. In particular, we do not consider \sh{pattern} which supports complex regular expressions, and the \sh{lessThanOrEquals} or \sh{lessThan} that are binary relations (the \lO\ and \lO' components of our grammar). All other features defined as filters in Sec.~\ref{sec:language} are represented by monadic relations $F(\vx)$ of the \SCL\ grammar.  

The actual problem imposed by filters w.r.t.\ deciding satisfiability and containment is that each combination of filters might be satisfied by a limited number of elements (zero, if the combination is unsatisfiable). For example, 
the number of elements of datatype boolean is two, the number of elements that are literals is infinite and the number of elements of datatype integer that are greater than 0 and lesser than 5 is four.  

Let a \emph{filter combination} $\mathds{F}(\vx)$ denote a conjunction of atoms of the form $\vx = \conc$, $\vx \neq \conc$, $F(\vx)$ or $\neg F(\vx)$, where \conc\ is a constant and $F$ is a filter predicate. 
Given a filter combination, it is possible to compute the number of elements that can satisfy it. Let $\gamma$ be the function from filter combinations to naturals returning this number. The computation of $\gamma(\mathds{F}(\vx))$ for the monadic filters we consider is trivial as it boils down to determining: (1) the lexical space and compatibility of datatypes and node types (including those implied by language tag and order constraints); (2) the cardinality of intervals defined by order or string-length constraints; and (3) simple RDF-specific restrictions, e.g., the fact that each node has at most one datatype and language tag. Combinations of the previous three points are equally computable. 
Let $\mathds{F}^{\varphi}$ be the set of filter combinations that can be constructed with the filters and constants occurring in a sentence $\varphi$.  
The filter axiomatization $\axiomatisation{\varphi}$ of a sentence $\varphi$ is
the following conjunction (conjuncts where $\gamma(\mathds{F}(\vx))$ is infinite
are trivially simplified to $\top$).
\[
\axiomatisation{\varphi} =  \bigwedge_{\mathds{F}(\vx) \in \mathds{F}^{\varphi}} \exists^{ \leq \gamma(\mathds{F}(\vx))} \vx . \; \mathds{F}(\vx)
\]

\begin{theorem} \label{thm:filters}
An \SCL\ sentence $\phi$ is satisfiable on a canonical model iff $\phi \wedge \axiomatisation{\phi}$ is satisfiable on an uninterpreted model.
Containment $\phi_1 \subseteq \phi_2$ of two \SCL\ sentences on all canonical models holds iff $\phi_1 \wedge \axiomatisation{\phi_1 \wedge \phi_2} \subseteq \phi_2$ holds on all uninterpreted models.
\end{theorem}
\begin{proof}[Proof sketch]
  We focus on satisfiability, since the proof for containment is similar.
  First notice that every canonical model $I$ of $\varphi$ is necessarily a
  model of $\phi \wedge \axiomatisation{\phi}$.
  Indeed, by definition of the function $\gamma$, given a filter combination
  $\mathds{F}(\vx)$, there cannot be more than $\gamma(\mathds{F}(\vx))$
  elements satisfying $\mathds{F}(\vx)$, independently of the underlying
  canonical model.
  Thus, $I$ satisfies $\axiomatisation{\phi}$.
  Consider now a model $I$ of $\phi \wedge \axiomatisation{\phi}$ and let
  $I^{\star}$ be the structure obtained from $I$ by replacing the
  interpretations of the monadic filter relations with their canonical ones.
  Obviously, for any filter combination $\mathds{F}(\vx)$, there are exactly
  $\gamma(\mathds{F}(\vx))$ elements in $I^{\star}$ satisfying
  $\mathds{F}(\vx)$, since $I^{\star}$ is canonical.
  As a consequence, there exists a injection $\iota$ between the elements
  satisfying $\mathds{F}(\vx)$ in $I$ and those satisfying $\mathds{F}(\vx)$ in
  $I^{\star}$.
  At this point, one can prove that $I^{\star}$ satisfies $\varphi$.
  Indeed, every time a value $x$, satisfying $\mathds{F}(\vx)$ in $I$, is used
  to verify a subformula $\psi$ of $\varphi$ in $I$, one can use the value
  $\iota(x)$ to verify the same subformula $\psi$ in $I^{\star}$.
\end{proof}

%% file: 4satisfiability.tex
\section{\SCL\ Satisfiability}\label{sec:scl-sat}

  \begin{figure}[t]
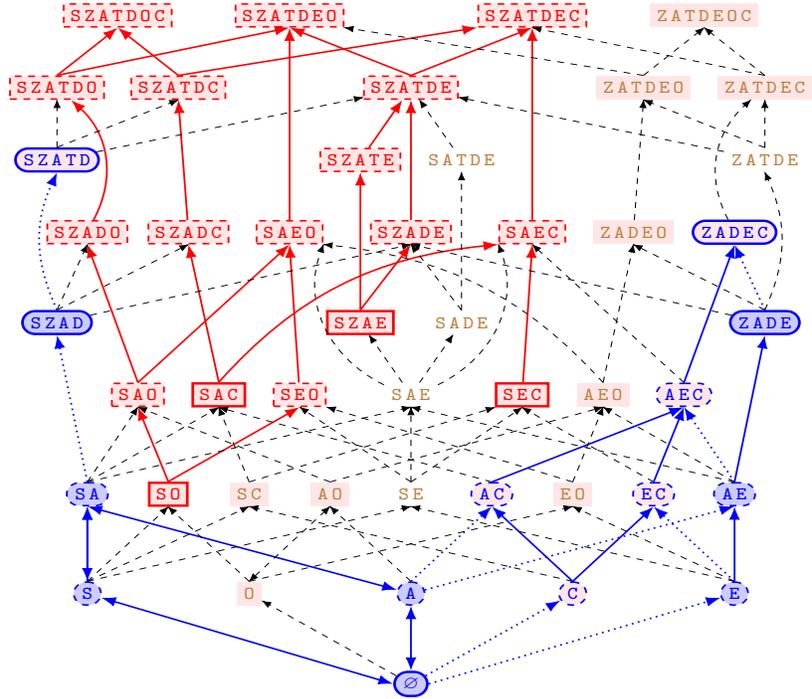

    \figfrg
    \caption{\label{fig:log} Decidability and complexity map of \SCL.
      Round (blue) and square (red) nodes denote decidable and undecidable
      fragments, respectively.
      Solid borders on nodes correspond to theorems in this paper, while dashed
      borders are implied results.
      Directed edges indicate inclusion of fragments, while bidirectional edges
      denote polynomial-time reducibility.
      Solid edges are preferred derivations to obtain tight results, while
      dotted ones leads to worst upper-bounds or model-theoretic properties.
      Finally, a light blue background indicates that the fragment enjoys the
      finite-model property, while those with a light red background do not
      satisfy this property.}
  \end{figure}

  In this section we embark on a detailed analysis of the satisfiability problem
  for different fragments of \SCL. Some of the proven and derived results are
  visualized in Figure~\ref{fig:log}.
  The decidability results are proved via embedding into known decidable
  (extensions of) fragments of first-order logic, while the undecidability ones
  are obtained through reductions from the domino problem.
  Since we are not considering filters explicitly, but via axiomatization, the
  only interpreted relations are the equality and the orderings.

  For the sake of space and readability, the map depicted in the figure is not
  complete \wrt two aspects.
  First, it misses few fragments whose decidability can be immediately derived
  via inclusion into a more expressive decidable fragment, e.g.,
  \Z\,\A\,\D\,\E\,\C\, or \S\,\Z\,\A\,\T\,\D.
  Second, the rest of the missing cases have an open decidability problem.
  In particular, while there are several decidable fragments containing the \T\
  feature we do not know any decidable fragment with the \lO\ or \lO' feature.
  Notice that the undecidability results making use of the last two are only
  applicable to generalized RDF.

  As first result, we show that the base language \X\ is already powerful enough
  to express properties writable by combining the \S, \Z, and \A\ features.
  In particular, the latter one does not augment the expressiveness when the \D\
  and \lO\ features are considered alone.

  \begin{theorem}
    \label{thm:eqv}
    There are semantic-preserving and polynomial-time finite-model-invariant
    satisfiability-preserving translations between the following \SCL\
    fragments:
    \begin{inparaenum}
      \item\label{thm:eqv(sza)}
        $\varnothing \equiv \S \equiv \Z \equiv \A \equiv \S\,\Z \equiv \S\,\A
        \equiv \Z\,\A \equiv \S\,\Z\,\A$;
      \item\label{thm:eqv(ad)}
        $\D \equiv \A\,\D$;
      \item\label{thm:eqv(ao)}
        $\lO \equiv \A\,\lO$;
      \item\label{thm:eqv(ado)}
        $\D\,\lO \equiv \A\,\D\,\lO$.
    \end{inparaenum}
  \end{theorem}
  \begin{proof}
    To show the equivalences between the fourteen \SCL\ fragments mentioned in
    the thesis, we consider the following first-order formula equivalences that
    represent few distributive properties enjoyed by the \S, \Z, and \A\
    features \wrt some of the other language constructs.
    The verification of their correctness only requires the application of
    standard properties of Boolean connectives and first-order quantifiers.
    \begin{itemize}
      \item
        \textbf{[\S].}
        The sequence combination of two path formulas $\pi_{1}$ and $\pi_{2}$ in
        the body of an existential quantification is removed by nesting two
        quantifications, one for each $\pi_{i}$:
        \[
          \exists \vy \,.\, (\exists \vz \,.\, \pi_{1}(\vx, \vz) \wedge
          \pi_{2}(\vz, \vy)) \wedge \psi(\vy) \equiv \exists \vz \,.\,
          \pi_{1}(\vx, \vz) \wedge (\exists \vy \,.\, \pi_{2}(\vz, \vy) \wedge
          \psi(\vy)).
        \]
      \item
        \textbf{[\Z].}
        The \Z\ path construct can be removed from the body of an existential
        quantification on a free variable $\vx$ by verifying whether the formula
        $\psi$ in its scope is already satisfied by the value bound to $\vx$
        itself:
        \[
          \exists \vy \,.\, (\vx = \vy \vee \pi(\vx, \vy)) \wedge \psi(\vy)
          \equiv \psi(\vx) \vee \exists \vy \,.\, \pi(\vx, \vy) \wedge
          \psi(\vy).
        \]
      \item
        \textbf{[\A].}
        The removal of the \A\ path construct from the body of an existential
        quantifier or of the \D\ and \lO\ constructs can be done by exploiting
        the following equivalences:
        \[
        \hspace{-3.5em}
        \begin{footnotesize}
        \begin{aligned}
          & \exists \vy \,.\, (\pi_{1}(\vx, \vy) \vee \pi_{2}(\vx, \vy)) \wedge
            \psi(\vy) \equiv
            (\exists \vy \,.\, \pi_{1}(\vx, \vy) \wedge \psi(\vy)) \vee (\exists
            \vy \,.\, \pi_{2}(\vx, \vy) \wedge \psi(\vy)); \\
          & \neg \exists \vy .\, (\pi_{1}(\vx, \vy) \vee \pi_{2}(\vx, \vy))
            \wedge R(\vx, \vy) \equiv
            (\neg \exists \vy .\, \pi_{1}(\vx, \vy) \wedge R(\vx, \vy))
            \!\wedge\! (\neg \exists \vy .\, \pi_{2}(\vx, \vy) \wedge R(\vx,
            \vy)); \\
          & \forall \vy, \vz \,.\, (\pi_{1}(\vx, \vy) \vee \pi_{2}(\vx, \vy))
            \wedge R(\vx, \vz) \rightarrow \sigma(\vy, \vz) \equiv
            (\forall \vy, \vz \,.\, \pi_{1}(\vx, \vy) \wedge R(\vx, \vz)
            \rightarrow \sigma(\vy, \vz)) \\
          & \hphantom{\forall \vy, \vz \,.\, (\pi_{1}(\vx, \vy) \vee
            \pi_{2}(\vx, \vy)) \wedge R(\vx, \vz) \rightarrow \sigma(\vy, \vz)}
            \,\;\wedge (\forall \vy, \vz \,.\, \pi_{2}(\vx, \vy) \wedge R(\vx,
            \vz) \rightarrow \sigma(\vy, \vz)).
        \end{aligned}
        \end{footnotesize}
        \]
    \end{itemize}
    At this point, the equivalences between the fragments naturally follow by
    iteratively applying the discussed equivalences.
    \\\indent
    The removal of the \Z\ and \A\ constructs from an existential quantification
    might lead, however, to an exponential blow-up in the size of the formula
    due to the duplication of the body $\psi$ of the quantification.
    To obtain polynomial-time finite-model-invariant
    satisfiability-preserving translations, we first construct from the given
    sentence $\varphi$ a finite-model-invariant equisatisfiable sentence
    $\varphi^{\star}$. The latter has a linear size in the original one and all
    the bodies of its quantifications are just plain relations.
    Then, we apply the above described semantic-preserving translations to
    $\varphi^{\star}$ that, in the worst case, only leads to a doubling in the
    size.
    The sentence $\varphi^{\star}$ is obtained by iteratively applying to
    $\varphi$ the following two rewriting operations, until no complex formula
    appears in the scope of an existential quantification.
    Let $\psi'(\vx) = \exists \vy \,.\, \pi(\vx, \vy) \wedge \psi(\vy)$ be a
    subformula, where $\psi(\vy)$ does not contain quantifiers other than
    possibly those of the \S, \D, and \lO\ features.
    Then:
    \begin{inparaenum}[(i)]
      \item
        replace $\psi'(\vx)$ with $\exists \vy \,.\, \pi(\vx, \vy) \wedge
        \hasshape{\vy}{\sconst}$, where $\sconst$ is a fresh
        constant;
      \item
        conjoin the resulting sentence with $\forall \vx .\;
        \hasshape{\vx}{\sconst} \leftrightarrow \psi(\vx)$.
    \end{inparaenum}
    The two rewriting operations only lead to a constant increase of
    the size and are applied only a linear number of times. 
  \end{proof}

  It turns out that the base language \X\ resembles the description logic
  $\mathcal{ALC}$ extended with universal roles, inverse roles, and
  nominals~\cite{BCMNP03}.
  This resemblance is exploited as the key observation at the core of the
  following result.

  \begin{theorem}
    \label{thm:sat(sza)}
    All \SCL\ subfragments of \S\,\Z\,\A\ enjoy the finite-model property.
    Moreover, the satisfiability problem is ExpTime-complete.
  \end{theorem}
  \begin{proof}
    The finite-model property follows from the fact that
    Theorem~\ref{thm:sat(szatd)} states the same property for the subsuming
    language \S\ \Z\ \A\ \D.
    As far as the satisfiability problem is concerned, thanks to
    Item~\ref{thm:eqv(sza)} of Theorem~\ref{thm:eqv}, we can focus on the base
    language \X.
    It can be observed that the description logic $\mathcal{ALC}$ extended with
    inverse roles and nominals~\cite{BCMNP03} and the language \X\ deprived of
    the universal quantifications at the level of sentences are linearly
    interreducible.
    Indeed, every existential modality $\exists R.C$ (\emph{resp.}, $\exists
    R^{-}.C$) precisely corresponds to the \SCL\ construct $\exists \vy \,.\,
    R(\vx, \vy) \wedge \psi_{C}(\vy)$ (\emph{resp.}, $\exists \vy \,.\,
    R^{-}(\vx, \vy) \wedge \psi_{C}(\vy)$), where $\psi_{C}(\vy)$
    represents the concept $C$.
    Moreover, every nominal $n$ corresponds to the equality construct $\vx =
    \conc_{n}$, where a natural bijection between nominals and constants is
    considered.
    Since the aforementioned description logic has an ExpTime-complete
    satisfiability problem~\cite{SCh91,DM00}, it holds that the same problem for
    all subfragments of \S\,\Z\,\A\ is ExpTime-hard.
    Completeness follows by observing that the universal quantifications at the
    level of sentences can be encoded in the further extension of
    $\mathcal{ALC}$ with the universal roles~\cite{SCh91}, which has an
    ExpTime-complete satisfiability problem~\cite{SV01}.
  \end{proof}

  To derive properties of the \Z\,\A\,\D\,\E\ fragment, together with its
  sub-fragments (two of those -- \E\ and \A\,\E\ -- are shown in
  Figure~\ref{fig:log}), we leverage on the syntactic embedding into the
  two-variable fragment of first-order logic.

  \begin{theorem}
    \label{thm:sat(zade)}
    The \Z\,\A\,\D\,\E\; fragment of \SCL\ enjoys the finite-model property.
    Moreover, the associated satisfiability problem is solvable in NExpTime.
  \end{theorem}
  \begin{proof}
    Via inspection of the \SCL\ grammar one can notice that, by avoiding the
    \S\ and \lO\ features of the language it is only possible to write formulas
    with at most two free variables~\cite{Mor75}.
    For this reason, every \Z\,\A\,\D\,\E\ formula belongs to the two-variable
    fragment of first-order logic which is known to enjoy both the finite-model
    property and a NExpTime satisfiability problem~\cite{GKV97}.
  \end{proof}

  The embedding used in the previous theorem can be generalized when the \C\
  feature is added to the picture.
  However, this additional expressive power does not come without a price
  since the complexity increases and the finite-model property is lost.

  \begin{theorem}
    \label{thm:sat(zadec)}
    The \C\, fragment of \SCL\ does not enjoy the finite-model property and has
    a NExpTime-hard satisfiability problem.
    Nevertheless, the finite and unrestricted satisfiability problems for
    \Z\,\A\,\D\,\E\,\C\, are NExpTime-complete.
  \end{theorem}
   \begin{proof}
    As for the proof of Theorem~\ref{thm:sat(zade)}, one can observe that every
    \Z\,\A\,\D\,\E\,\C\ formula belongs to the two-variable fragment of
    first-order logic extended with counting quantifiers.
    Such a logic does not enjoy the finite-model
    property~\cite{GOR97}, since it syntactically contains a sentence that
    encodes the existence of an injective non-surjective function from the
    domain of the model to itself.
    The \C\ fragment of \SCL\ allows us to express the same property via the
    following sentence $\varphi$, thus implying the first part of the
    thesis:
    \[
    \begin{aligned}
      \varphi \grameq\;
        & \isA(0, \conc) \wedge \neg \exists \vx \,.\, R^{-}(0, \vx) \wedge
          \forall \vx \,.\, \isA(\vx, \conc) \rightarrow \psi(\vx); \\
      \psi(\vx) \grameq\;
        & \exists^{= 1} \vy \,.\, (R(\vx, \vy) \wedge \isA(\vy, \conc)) \wedge
          \neg \exists^{\geq 2} \vy \,.\, R^{-}(\vx, \vy).
    \end{aligned}
    \]
    Intuitively, the first two conjuncts of $\varphi$ force every model of the
    sentence to contain an element $0$ that does not have any $R$-predecessor
    and that is related to $\conc$ in the $\isA$ relation.
    In other words, $0$ is not contained in the image of the relation $R$.
    The third conjunct of $\varphi$ ensures that every element related to
    $\conc$ \wrt $\isA$ has exactly one $R$-successor, also related to $\conc$
    in the same way, and at most one $R$-predecessor.
    Thus, a model of $\varphi$ must contain an infinite chain of elements
    pairwise connected by the functional relation $R$.


    By generalizing the proof of Theorem~\ref{thm:sat(sza)}, one can notice that
    the \C\ fragment of \SCL\ semantically subsumes the description logic
    $\mathcal{ALC}$ extended with inverse roles, nominals, and cardinality
    restrictions~\cite{BCMNP03}.
    Indeed, every qualified cardinality restriction $(\geq n\, R.C)$
    (\emph{resp.}, $(\leq n\, R.C)$) precisely corresponds to the \SCL\
    construct $\exists^{\geq n} \vy \,.\, R(\vx, \vy) \wedge \psi_{C}(\vy)$
    (\emph{resp.}, $\neg \exists^{\geq n + 1} \vy \,.\, R(\vx, \vy) \wedge
    \psi_{C}(\vy)$), where $\psi_{C}(\vy)$ represents the concept $C$.
    Thus, the hardness result for \C\ follows by recalling that the specific
    $\mathcal{ALC}$ language has a NExpTime-hard satisfiability
    problem~\cite{Tob00,Lut04}.
    On the positive side, however, the extension of the two-variable fragment of
    first-order logic with counting quantifiers has decidable finite and
    unrestricted satisfiability problems.
    Specifically, both can be solved in NExpTime, even in the case of binary
    encoding of the cardinality constants~\cite{Pra05,Pra10}.
    Hence, the second part of the thesis follows as well.
  \end{proof}

  Thanks to the axiomatization of (the subset of) filters given in
  Sec.~\ref{sec:filters}, it is immediate to see that the \Z \A \D \E \C\
  fragment extended with these filters is decidable as well.
  Indeed, although the sentence $\axiomatisation{\varphi}$ is not immediately
  expressible in \SCL\ it belongs to the two-variable fragment of FOL extended
  with counting quantifiers.
  Notice however that, since $\axiomatisation{\varphi}$ might be exponential in
  the size of $\varphi$, this approach only leads to a (potentially) coarse upper
  bound. 
  An attempt to prove a tight complexity result might exploit the SMT-like
  approach described in~\cite{ABM19} for the LTL part of Strategy Logic.
  Indeed, one could think to extend the decision procedure for the above FOL fragment
  in such a way that the filter axiomatization is implicitly considered during the
  check for satisfiability.

  For the \S\ \Z\ \A\ \D\ fragment, we obtain model-theoretic and complexity
  results via an embedding into the unary-negation fragment of first-order
  logic.
  When the \T\ feature is considered, the same embedding can be adapted to
  rewrite \S\ \Z\ \A\ \T\ \D\ into the extension of the above first-order
  fragment with regular path expressions.
  Unfortunately, as for the addition of the \C\ feature to \Z\,\A\,\D\,\E, we
  pay the price of losing the finite-model property.
  In this case, however, no increase of the complexity of the satisfiability
  problem occurs.

   \begin{theorem}
    \label{thm:sat(szatd)}
    The \S\,\Z\,\A\,\D\, fragment of \SCL\ enjoys the finite-model property.
    The \S\,\T\,\D\, fragment does not enjoy the finite-model property.
    However, the finite and unrestricted satisfiability problems for
    \S\,\Z\,\A\,\T\,\D\, are solvable in 2ExpTime.
  \end{theorem}
  \begin{proof}
    By inspecting the \SCL\ grammar, one can notice that every formula that does
    not make use of the \T, \E, \lO, and \C\ constructs can be translated into
    the standard first-order logic syntax, with conjunctions and disjunctions as
    unique binary Boolean connectives, where negation is only applied to
    formulas with at most one free variable.
    For this reason, every \S\,\Z\,\A\,\D\ formula semantically belongs to the
    unary-negation fragment of first-order logic, which is known to enjoy the
    finite-model property~\cite{CS11,CS13}.

    Similarly every \S\,\Z\,\A\,\T\,\D\ formula belongs to the unary-negation
    fragment of first-order logic extended with regular path
    expressions~\cite{JLMS18}.
    Indeed, the grammar rule $\pi(\vx, \vy)$ of \SCL, precisely resembles the
    way the regular path expressions are constructed in the considered logic,
    when one avoids the test construct.
    Unfortunately, as for the two-variable fragment with counting quantifiers, this logic also fails to satisfy the finite-model property since it is able
    to encode the existence of a non-terminating path without cycles.
    The \S\,\T\,\D\ fragment of \SCL\ allows us to express the same property, as
    described in the following.
    First of all, consider the \S\,\T\ path formula $\pi(\vx, \vy) \grameq
    \exists \vz \,.\, (R^{-}(\vx, \vz) \wedge (R^{-}(\vz, \vy))^{\star})$.
    Obviously, $\pi(\vx, \vy)$ holds between two elements $\vx$ and $\vy$ of a
    model if and only if there exists a non-trivial $R$-path (of arbitrary
    positive length) that, starting in $\vy$, leads to $\vx$.
    Now, by writing the \S\,\T\,\D\ formula $\psi(\vx) \grameq \neg \exists \vy
    \,.\, (\pi(\vx, \vy) \wedge R(\vx, \vy))$, we express the fact that an
    element $\vx$ does not belong to any $R$-cycle since, otherwise, there
    would be an $R$-successor $\vy$ able to reach $\vx$ itself.
    Thus, by ensuring that every element in the model has an $R$-successor, but
    does not belong to any $R$-cycle, we can enforce the existence of an
    infinite $R$-path.
    The \S\,\T\,\D\ sentence $\varphi$ expresses exactly this property:
    \[
      \varphi \grameq \isA(0, \conc) \wedge \forall \vx \,.\, \isA(\vx, \conc)
      \rightarrow (\psi(\vx) \wedge \exists \vy \,.\, (R(\vx, \vy) \wedge
      \isA(\vy, \conc))).
    \]


    On the positive side, however, the extension of the unary-negation fragment
    of first-order logic with arbitrary transitive relations or, more generally,
    with regular path expressions has decidable finite and unrestricted
    satisfiability problems.
    Specifically, both can be solved in 2ExpTime~\cite{ABBV16,JLMS18,DK19}.
  \end{proof}


  At this point, it is interesting to observe that the \lO\ feature allows us to
  express a very weak form of counting restriction which is, however, powerful
  enough to lose the finite-model property. For the proof of the following we refer to our appendix.

  \begin{theorem}
    \label{thm:fmp(oeop)}
    \SCL\ fragments \lO\ and \E\,\lO' do not satisfy the finite-model property.
  \end{theorem}

  In the remaining part of this section, we show the undecidability of the
  satisfiability problem for five fragments of \SCL\ through a semi-conservative
  reduction from the standard domino problem~\cite{Wan61,Ber66,Rob71}, whose
  solution is known to be $\Pi_{0}^{1}\text{-complete}$.
  A $\mathbb{N} \times \mathbb{N}$ tiling system $(\TileSet, \HorizontalRel,
  \VerticalRel)$ is a structure built on a non-empty set $\TileSet$ of domino
  types, a.k.a. tiles, and two horizontal and vertical matching relations
  $\HorizontalRel, \VerticalRel \subseteq \TileSet \times \TileSet$.
  The domino problem asks for a compatible tiling of the first quadrant
  $\mathbb{N} \times \mathbb{N}$ of the plane, \ie, a solution mapping $\eth
  \colon \mathbb{N} \times \mathbb{N} \to \TileSet$ such that, for all $\vx, \vy
  \in \mathbb{N}$, both $(\eth(\vx, \vy), \eth(\vx + 1, \vy)) \in
  \HorizontalRel$ and $(\eth(\vx, \vy), \eth(\vx, \vy + 1)) \in \VerticalRel$
  hold true.

  \begin{theorem}
    \label{thm:sat(und)}
    The satisfiability problems of the \S\,\lO, \S\,\A\,\C, \S\,\E\,\C,
    \S\,\E\,\lO', and \S\,\Z\,\A\,\E\, fragments of \SCL\ are undecidable.
  \end{theorem}
  \begin{proof}
    The main idea behind the proof is to embed a tiling system into a model of a
    particular \SCL\ sentence that is satisfiable if and only if the tiling
    system allows for an admissible tiling.
    The hardest part in the reduction consists in the definition of a
    satisfiable sentence all of whose models homomorphically contain the
    infinite grid of the tiling problem. In other words, this sentence should admit an
    infinite square grid graph as a minor of the model unwinding.
    Given that, the remaining part of the reduction can be completed in the base
    language \X.

    Independently of the fragment we are proving undecidable, consider the
    sentence
    \[
      \varphi \grameq \bigvee_{\tile \in \TileSet} \isA(0, \tile) \wedge
      \bigwedge_{\tile \in \TileSet} \forall \vx \,.\, \isA(\vx, \tile)
      \rightarrow (\psi_{T}^{\tile}(\vx) \wedge \psi_{G}(\vx)).
    \]
    Intuitively, this first states the existence of the point $0$, the origin of
    the grid, labeled by some tile and then ensures the fact that all points
    $\vx$, that are labeled by some tile $\tile$, need to satisfy the two
    formulas $\psi_{T}^{\tile}(\vx)$ and $\psi_{G}(\vx)$.
    The first formula is used to ensure the admissibility of the tiling, while the
    second one forces the model to embed a grid.

    \[
    \hspace{-1.5em}
    \begin{aligned}
      \psi_{T}^{\tile}(\vx) \grameq\;
        & \bigwedge_{\tileprime \in \TileSet}^{\tileprime \neq \tile} \neg
          \isA(\vx, \tileprime) \\
      \wedge\;
        & \left(\forall \vy \,.\, \Horizontal(\vx, \vy) \rightarrow
          \bigvee_{(\tile, \tileprime) \in \HorizontalRel} \isA(\vy, \tileprime)
          \right) \wedge \left(\forall \vy \,.\, \Vertical(\vx,\vy) \rightarrow
          \bigvee_{(\tile, \tileprime) \in \VerticalRel} \isA(\vy, \tileprime)
          \right)
    \end{aligned}
    \]
    The first conjunct of the \X\ formula $\psi_{T}^{\tile}(\vx)$ verifies that
    the point $\vx$ is labeled by no other tile than $\tile$.
    The second part, instead, ensures that the points $\vy$ on the right or
    above of $\vx$ are labeled by some tile $\tile'$ which is compatible with
    $\tile$, \wrt the constraints imposed by the horizontal $\HorizontalRel$ and
    vertical $\VerticalRel$ relations, respectively.

    At this point, we can focus on the formula $\psi_{G}(\vx)$ defined as
    follows:
    \[
      \psi_{G}(\vx) \grameq (\exists \vy \,.\, \Horizontal(\vx, \vy)) \wedge
      (\exists \vy \,.\, \Vertical(\vx, \vy)) \wedge \gamma(\vx).
    \]
    The first two conjuncts guarantee the existence of an horizontal and
    vertical adjacent of the point $\vx$, while the subformula $\gamma(\vx)$,
    whose definition depends on the considered fragment of \SCL, needs to
    enforce the fact that $\vx$ is the origin of a square. That is, that going
    horizontally and then vertically or, vice versa, vertically and then
    horizontally the same point is reached.
    In order to do this, we make use of the two \S\ path formulas
    $\pi_{\Horizontal\Vertical}(\vx, \vy) \grameq\exists \vz \,.\,
    (\Horizontal(\vx, \vz) \wedge \Vertical(\vz, \vy))$ and
    $\pi_{\Vertical\Horizontal}(\vx, \vy) \grameq \exists \vz \,.\,
    (\Vertical(\vx, \vz) \wedge \Horizontal(\vz, \vy))$.
    In some cases, we also use the \S\,\A\ path formula $\pi_{\Diagonal}(\vx,
    \vy) \grameq \pi_{\Horizontal\Vertical}(\vx, \vy) \vee
    \pi_{\Vertical\Horizontal}(\vx, \vy)$ combining the previous ones.
    We now proceed by a case analysis on the specific fragments.
    \begin{itemize}
      \item
        \textbf{[\S\,\O]}
        By assuming the existence of a non-empty relation $\Diagonal$ connecting
        a point with its opposite in the square, \ie, the diagonal point, we
        can say that all points reachable through $\pi_{\Horizontal\Vertical}$
        or $\pi_{\Vertical\Horizontal}$ are, actually, the same unique point:
        \[
        \hspace{-2em}
        \begin{aligned}
          \gamma(\vx) \grameq\;
            & \exists \vy \,.\, \Diagonal(\vx, \vy) \\
          \wedge\;
            & \forall \vy, \vz .\, \pi_{\Horizontal\Vertical}(\vx, \vy) \wedge
              \Diagonal(\vx, \vz) \rightarrow \vy \leq \vz \;\wedge\; \forall
              \vy, \vz .\, \pi_{\Horizontal\Vertical}(\vx, \vy) \wedge
              \Diagonal(\vx, \vz) \rightarrow \vy \geq \vz \\
          \wedge\;
            & \forall \vy, \vz .\, \pi_{\Vertical\Horizontal}(\vx, \vy) \wedge
              \Diagonal(\vx, \vz) \rightarrow \vy \leq \vz \;\wedge\; \forall
              \vy, \vz .\, \pi_{\Vertical\Horizontal}(\vx, \vy) \wedge
              \Diagonal(\vx, \vz) \rightarrow \vy \geq \vz.
        \end{aligned}
        \]
        The \S\,\O\ formula $\gamma(\vx)$ ensures that
        relation $\Diagonal$ is non-empty and functional and that all points
        reachable via $\pi_{\Horizontal\Vertical}$ or
        $\pi_{\Vertical\Horizontal}$ are necessarily the one reachable through
        $\Diagonal$.
      \item
        \textbf{[\S\,\A\,\C]}
        By applying a counting quantifier to the formula $\pi_{\Diagonal}$
        encoding the union of the points reachable through
        $\pi_{\Horizontal\Vertical}$ or $\pi_{\Vertical\Horizontal}$, we can
        ensure the existence of a single diagonal point:
        $\gamma(\vx) \grameq \neg \exists^{\geq 2} \vy \,.\,
        \pi_{\Diagonal}(\vx, \vy)$.
      \item
        \textbf{[\S\,\E\,\C]}
        As for the \S\,\O\ fragment, here we use a diagonal relation
        $\Diagonal$, which needs to contain all and only the points reachable
        via $\pi_{\Horizontal\Vertical}$ or $\pi_{\Vertical\Horizontal}$.
        By means of the counting quantifier, we ensure its functionality:
        \[
        \begin{aligned}
          \gamma(\vx) \grameq
             \neg \exists^{\geq 2} \vy . \, \Diagonal(\vx, \vy) 
          \wedge
             \forall \vy .  \pi_{\Horizontal\Vertical}(\vx, \vy)
              \leftrightarrow \Diagonal(\vx, \vy) \wedge
             \forall \vy .  \pi_{\Vertical\Horizontal}(\vx, \vy)
              \leftrightarrow \Diagonal(\vx, \vy).
        \end{aligned}
        \]
      \item
        \textbf{[\S\,\E\,\O']}
        This case is similar to the previous one, where the functionality of
        $\Diagonal$ is obtained by means of the \O\ construct:
        \[
        \begin{aligned}
          \gamma(\vx) \grameq\;
            & \forall \vy, \vz \,.\, \Diagonal(\vx, \vy) \wedge \Diagonal(\vx,
              \vz) \rightarrow \vy \leq \vz \\
          \wedge\;
            & \forall \vy \,.\, \pi_{\Horizontal\Vertical}(\vx, \vy)
              \leftrightarrow \Diagonal(\vx, \vy) \:\wedge\; \forall \vy \,.\,
              \pi_{\Vertical\Horizontal}(\vx, \vy) \leftrightarrow
              \Diagonal(\vx, \vy).
        \end{aligned}
        \]
      \item
        \textbf{[\S\,\Z\,\A\,\E]}
        This proof 
        is inspired by the one used for the undecidability of the guarded fragment
        extended with transitive closure~\cite{Gra99a}.
        This time, the functionality of the diagonal relation $\Diagonal$ is
        indirectly ensured by the conjunction of the four formulas
        $\gamma_{1}(\vx)$, $\gamma_{2}(\vx)$, $\gamma_{3}(\vx)$, and
        $\gamma_{4}(\vx)$ that exploit all the features of
        the fragment:
        \[
        \begin{aligned}
          \gamma(\vx) \grameq
            & \; \gamma_{1}(\vx) \wedge \gamma_{2}(\vx) \wedge \gamma_{3}(\vx)
              \wedge \gamma_{4}(\vx) \wedge \forall \vy \,.\,
              \pi_{\Diagonal}(\vx, \vy) \leftrightarrow \Diagonal(\vx, \vy),
         \end{aligned}
        \]    
        where
        \[
        \begin{aligned}
          \gamma_{1}(\vx) \grameq
          & \;\forall \vy \,.\, \left( \bigvee_{i \in \{ 0, 1 \}}
            \DiagonalX{i}(\vx, \vy) \right) \leftrightarrow \Diagonal(\vx, \vy),
            \\
          \gamma_{2}(\vx) \grameq
          & \left( \bigvee_{i \in \{ 0, 1 \}} \neg \exists \vy .\,
            \DiagonalX{i}(\vx, \vy) \right) \!\wedge\! \left( \bigwedge_{i \in \{ 0,
            1 \}} \forall \vy .\, \DiagonalX{i}(\vx, \vy) \rightarrow \exists
            \vz .\, \DiagonalX{1 - i}(\vy, \vz) \right)\!\!, \\
          \gamma_{3}(\vx) \grameq
          & \bigwedge_{i \in \{ 0, 1 \}} \forall \vy \,.\, \left( \vx = \vy \vee
            \DiagonalX{i}(\vx, \vy) \vee \DiagonalX{i}^{-}(\vx, \vy) \right)
            \leftrightarrow \Closure{i}(\vx, \vy), \text{ and} \\
          \gamma_{4}(\vx) \grameq
          & \bigwedge_{i \in \{ 0, 1 \}} \forall \vy \,.\, (\exists \vz .
            (\Closure{i}(\vx, \vz) \wedge \Closure{i}(\vz, \vy)))
            \leftrightarrow \Closure{i}(\vx, \vy).
        \end{aligned}
        \]
        Intuitively, $\gamma_{1}$ asserts that $\Diagonal$ is the union of the
        two accessory relations $\DiagonalX{0}$ and $\DiagonalX{1}$, while
        $\gamma_{2}$ guarantees that a point can only have adjacents \wrt just
        one relation $\DiagonalX{i}$ and that these adjacents can only appear as
        first argument of the opposite relation $\DiagonalX{1 - i}$.
        In addition, $\gamma_{3}$ ensures that the additional relation
        $\Closure{i}$ is the reflexive symmetric closure of $\DiagonalX{i}$ and
        $\gamma_{4}$ forces $\Closure{i}$ to be transitive as well. 
        
        We can now prove that the relation $\Diagonal$ is functional.
        Suppose by contradiction that this is not case, \ie, there exist values
        $a$, $b$, and $c$ in the domain of the model of the sentence $\varphi$,
        with $b \neq c$ such that both $\Diagonal(a, b)$ and $\Diagonal(a, c)$
        hold true.
        By the formula $\gamma_{1}$ and the first conjunct of $\gamma_{2}$, we
        have that $\DiagonalX{i}(a, b)$ and $\DiagonalX{i}(a, c)$ hold for
        exactly one index $i \in \{ 0, 1 \}$.
        Thanks to the full $\gamma_{2}$, we surely know that $a \neq b$, $a \neq
        c$, and neither $\DiagonalX{i}(b, c)$ nor $\DiagonalX{i}(c, b)$ can
        hold.
        Indeed, if $a = b$ then $\DiagonalX{i}(a, a)$. This in turn implies
        $\DiagonalX{1 - i}(a, d)$ for some value $d$ due to the second
        conjunct of $\gamma_{2}$.
        Hence, there would be pairs with the same first element in both
        relations, trivially violating the first conjunct of $\gamma_{2}$.
        Similarly, if $\DiagonalX{i}(b, c)$ holds, then $\DiagonalX{1 - i}(c,
        d)$ needs to hold as well, for some value $d$, leading again to a
        contradiction.
        Now, by the formula $\gamma_{3}$, both $\Closure{i}(b, a)$ and
        $\Closure{i}(a, c)$ hold, but $\Closure{i}(b, c)$ does not.
        However, this clearly contradicts $\gamma_{4}$.
        As a consequence, $\Diagonal$ is necessarily functional.
    \end{itemize}
    Now, it is not hard to see that the above sentence $\varphi$ (one for each
    fragment) is satisfiable if and only if the domino instance on which the
    reduction is based on is solvable.
  \end{proof}

%% file: 9conclusion.tex
\section{Conclusion} 

In this paper we define and study the decision problems of satisfiability and containment for SHACL documents and shape constraints. In order to do so, we introduce a complete translation between SHACL and \SCL, a fragment of FOL extended with counting quantifiers and a transitive closure operator. 
Using these translations we lay out a map of SHACL fragments for which we are able to prove undecidability or decidability along with complexity results, for the satisfiability and containment problems.  
We also expose semantic properties and asymmetries within  SHACL  which might inform a future update of the specification. The satisfiability and containment problems are undecidable for the full SHACL specification. However, decidability can be achieved by restricting the usage of certain SHACL components, such as cardinality restrictions over property shapes or property paths.
Nevertheless, the decidability of some fragments of SHACL remains an open question, worthy of further investigation.  

%% file: zappendix.tex
\input{zappendixTranslationFromSHACL2FOL.tex}

\input{zappendixFOL2SHACL.tex}

\input{zappendixProof.tex}

%% file: zappendixTranslationFromSHACL2FOL.tex
\section{Translation from SHACL into \SCL\ Grammar}

We present our translation $\tau(\vShapeDocument)$ from a SHACL document \vShapeDocument\ (a set of SHACL shape definitions) into our \SCL\ grammar. 
The translation into \SCL\ grammar of a document \vShapeDocument\ containing a set $\vShapeDocument^{S}$ of shapes can be defined as: $\bigwedge_{\vs \in \vShapeDocument^{S}} \tau(\vs)$, where $\tau(\vs)$ is the translation of a single SHACL shape $\vs$. 
Notice that $\vShapeDocument$ contains an element for each shape name  occurring in $\vShapeDocument$. If $\vShapeDocument$ contains a shape name $\vs$ that does not have a corresponding shape definition, $\vShapeDocument^{S}$ will include the empty shape definition $\vs\pair{\{\}}{\{\}}$. 
Given a shape $\vs\pair{\vshapet}{\vshapec}$, its translation  $\tau(\vs\pair{\vshapet}{\vshapec})$ is defined in Table 1, where its constraint definition $\taus{\vshapec}{\vx}$ equals $\tau(\vx, \vshape)$. Note that we do not discuss implicit class-based targets, as they just represent a syntactic variant of class targets. In the reminder of this section we define how to compute $\tau(\vx, \vshape)$. 


As convention, we use \conc\ as an arbitrary constant and \constantList\ as an arbitrary list of constants. 
We use \vshape ,  \vshape$^{\prime}$\ and \vshape$^{''}$ as shape names, and \vShapes\  as a list of shape names. Variables are defined as $\vx$, $\vy$ and $\vz$. Arbitrary paths are identified with $\tr$.

The translation of the constraints of a shape $\tau(\vx, \vshape)$ is defined in two cases as follows. The first case deals with the property shapes, which must have exactly one value for the \sh{path} property. The second case deals with node shapes, which cannot have any value for the \sh{path} property.

\[
\tau(\vx, \vshape) = \top \wedge \begin{cases}
      \bigwedge_{\forall \trip{\vshape}{\vy}{\vz} \in \vShapeDocument} \tau_2(\vx, \tr, \trip{\vshape}{\vy}{\vz}) & \text{if} \; \exists r .    \trip{\vshape}{ \sh{path}}{ \tr}) \in \vShapeDocument\\
      \bigwedge_{\forall \trip{\vshape}{\vy}{\vz} \in \vShapeDocument} \tau_1(\vx, \trip{\vshape}{\vy}{\vz})& \text{otherwise}
    \end{cases}  
\]

This translation is based on the following translations of node shapes triples, property shape triples and property paths. 

\subsection{Translation of Node Shape Triples}

The translation of $\tau_1(\vx, \trip{\vshape}{\vy}{\vz})$ is split in the following cases, depending on the predicate of the triple. In case none of those cases are matched $\tau_1(\vx, \trip{\vshape}{\vy}{\vz}) \transeq \top$. The latter ensures that any triple not directly described in the cases below does not alter the truth value of the conjunction in the definition of $\tau(\vx, \vshape)$.

\begin{itemize}
    \item
        $\tau_1(\vx, \trip{\vshape}{ \sh{hasValue}}{ \conc}) \transeq 
        \vx = \conc $ .  
    \item
        $\tau_1(\vx, \trip{\vshape}{ \sh{in}}{\constantList}) \transeq \bigvee_{\conc \in \constantList}
        \vx = \conc $ .  
    \item
        $\tau_1(\vx, \trip{\vshape}{ \sh{class}}{ \conc}) \transeq 
        \exists \vy . \isA(\vx, \vy) \wedge \vy = \conc $ . 
    \item
        $\tau_1(\vx,\trip{\vshape}{ \sh{datatype}}{ \conc})) \transeq 
        \text{F}^{\, \hasdatatype = \conc}(\vx) $ . 
    \item
        $\tau_1(\vx, \trip{\vshape}{ \sh{nodeKind}}{ \conc}) \transeq 
        \text{F}^{\, \isIRI}(\vx) $ if $c = $\sh{IRI};
        $\text{F}^{\, \isLiteral}(\vx) $ if \\$c = $\sh{Literal};
        $\text{F}^{\, \isBlank}(\vx) $ if $c = $\sh{BlankNode}. The translations for a $\conc$ that equals \sh{BlankNodeOrIRI}, \sh{BlankNodeOrLiteral} or \sh{IRIOrLiteral} are trivially constructed by a conjunction of two of these three filters.
    \item
        $\tau_1(\vx, \trip{\vshape}{ \sh{minExclusive}}{ \conc}) \transeq \vx > \conc
        $ if order is an interpreted relation, else $\text{F}^{\, > \conc}(\vx)$. 
    \item
        $\tau_1(\vx, \trip{\vshape}{ \sh{minInclusive}}{ \conc}) \transeq \vx \geq \conc
         $ if order is an interpreted relation, else $\text{F}^{\, \geq \conc}(\vx)$. 
    \item
        $\tau_1(\vx, \trip{\vshape}{ \sh{maxExclusive}}{ \conc}) \transeq  \vx < \conc
         $ if order is an interpreted relation, else $\text{F}^{\, < \conc}(\vx)$.  
    \item
        $\tau_1(\vx, \trip{\vshape}{ \sh{maxInclusive}}{ \conc}) \transeq \vx \leq \conc
         $ if order is an interpreted relation, else $\text{F}^{\, \leq \conc}(\vx)$.  
    \item
        $\tau_1(\vx, \trip{\vshape}{ \sh{maxLength}}{ \conc}) \transeq 
        \text{F}^{\, \text{maxLength}=\conc}(\vx) $ . 
    \item
        $\tau_1(\vx, \trip{\vshape}{ \sh{minLength}}{ \conc}) \transeq 
        \text{F}^{\, \text{minLength}=\conc}(\vx) $ . 
    \item
        $\tau_1(\vx, \trip{\vshape}{ \sh{pattern}}{ \conc}) \transeq 
        \text{F}^{\, \text{pattern}=\conc}(\vx) $ . 
    \item
        $\tau_1(\vx, \trip{\vshape}{ \sh{languageIn}}{ \constantList}) \transeq \bigvee_{\conc \in \constantList} F^{\text{languageTag = \conc}}(\vx)$ . 
    \item
        $\tau_1(\vx, \trip{\vshape}{ \sh{not}}{ \vshape^{\prime}}) \transeq \neg \hasshape{\vx}{\vshape^{\prime}}$ .
    \item
        $\tau_1(\vx, \trip{\vshape}{ \sh{and}}{ \vShapes}) \transeq 
        \bigwedge_{\vshape^{\prime} \in \vShapes} \hasshape{\vx}{\vshape^{\prime}}$ . 
    \item
        $\tau_1(\vx, \trip{\vshape}{ \sh{or}}{ \vShapes}) \transeq 
        \bigvee_{\vshape^{\prime} \in \vShapes} \hasshape{\vx}{\vshape^{\prime}}$ . 
    \item
        $\tau_1(\vx, \trip{\vshape}{ \sh{xone}}{ \vShapes}) \transeq 
        \bigvee_{\vshape^{\prime} \in \vShapes} ( \hasshape{\vx}{\vshape^{\prime}}  \wedge  \bigwedge_{\vshape^{''} \in \vShapes \setminus \{\vshape^{\prime}\}} \\ \neg \hasshape{\vx}{\vshape^{''}})$ . 
    \item
        $\tau_1(\vx, \trip{\vshape}{ \sh{node}}{ \vshape^{\prime}}) \transeq 
        \hasshape{\vx}{\vshape^{\prime}}$ . 
    \item
        $\tau_1(\vx, \trip{\vshape}{ \sh{property}}{ \vshape^{\prime}}) \transeq 
        \hasshape{\vx}{\vshape^{\prime}}$ . 
\end{itemize}

\subsection{Translation of Property Shapes}

The translation of $\tau_2(\vx, \tr, \trip{\vshape}{\vy}{\vz})$ is split in the following cases, depending on the predicate of the triple. In case none of those cases are matched $\tau_2(\vx, \tr, \trip{\vshape}{\vy}{\vz}) \transeq \top$.

\begin{itemize}
    \item
         $\tau_2(\vx, \tr, \trip{\vshape}{ \sh{hasValue}}{ \conc}) \transeq 
         \exists \vy . r(\vx, \vy) \wedge \tau_1(\vy, \trip{\vshape}{ \sh{hasValue}}{ \conc}) $
    \item
        $\tau_2(\vx, \tr, \trip{\vshape}{ p}{ \conc}) \transeq 
        \forall \vy . \tau_3(\vx, \tr,\vy)) \rightarrow \tau_1(\vy, \trip{\vshape}{ p}{ \conc}) $, if $p$ equal to one of the following: \sh{class}, \sh{datatype}, \sh{nodeKind}, \sh{minExclusive}, \\ \sh{minInclusive}, \sh{maxExclusive}, \sh{maxInclusive}, \sh{maxLength}, \\ \sh{minLength}, \sh{pattern}, \sh{not}, \sh{and}, \sh{or}, \sh{xone},  \sh{node}, \\ \sh{property}, \sh{in} .
    \item
        $\tau_2(\vx, \tr, \trip{\vshape}{ \sh{languageIn}}{ \constantList}) \transeq 
        \forall \vy . \tau_3(\vx, \tr,\vy)) \rightarrow \\ \phantom{xxx} \tau_1(\vy, \trip{\vshape}{ \sh{languageIn}}{ \constantList}) $ .
    \item
        $\tau_2(\vx, \tr, \trip{\vshape}{ \sh{uniqueLang}}{ \text{true}}) \transeq \bigwedge_{\conc \in L} \neg \exists^{\geq 2} \vy . r(\vx, \vy) \wedge F^{\text{lang}=c}(\vy) $  where $L = \{\conc | \conc \in \constantList \wedge \exists \vshape^{'} .  \trip{\vshape^{'}}{\sh{languageIn}}{\constantList} \in \vShapeDocument\}$. This translation is possible because \sh{languageIn} is the only constraint that can force language tags constraints on literals.

        
    \item
        $\tau_2(\vx, \tr, \trip{\vshape}{ \sh{minCount}}{\conc}) \transeq 
        \exists^{\geq \conc} \vy . \tau_3(\vx, \tr,\vy)) $ .
    \item
        $\tau_2(\vx, \tr, \trip{\vshape}{ \sh{maxCount}}{ \conc}) \transeq \neg
        \exists^{\leq \conc} \vy . \tau_3(\vx, \tr,\vy)) $ .
    \item
        $\tau_2(\vx, \tr, \trip{\vshape}{ \sh{equals}}{ \conc}) \transeq \forall \vy .  \tau_3(\vx, \tr,\vy) \leftrightarrow \tau_3(\vx, \conc,\vy)$ .
    \item
        $\tau_2(\vx, \tr, \trip{\vshape}{ \sh{disjoint}}{ \conc}) \transeq \neg \exists \vy .  \tau_3(\vx, \tr,\vy) \wedge \tau_3(\vx, \conc,\vy)$ .
    \item
        $\tau_2(\vx, \tr, \trip{\vshape}{ \sh{lessThan}}{ \conc}) \transeq 
        \forall \vy, \vz \,.\, \tau_3(\vx, \tr,\vy) \wedge \tau_3(\vx, \conc,\vz) \rightarrow \vy < \vz $ .
    \item
        $\tau_2(\vx, \tr, \trip{\vshape}{ \sh{lessThanOrEquals}}{ \conc}) \transeq 
        \forall \vy, \vz \,.\, \tau_3(\vx, \tr,\vy) \wedge \tau_3(\vx, \conc,\vz) \rightarrow \\ \phantom{xxx} \vy \leq \vz $ .
    \item
        $\tau_2(\vx, \tr, \trip{\vshape}{ \sh{qualifiedValueShape}}{ \vshape^{\prime}}) \transeq \alpha \wedge \beta
        $ , where $\alpha$ and $\beta$ are defined as follows.        
        Let $S^{'}$ be the set of \emph{sibling shapes} of $\vshape$ if $\vShapeDocument$ contains \\ $\trip{\vshape}{  \sh{qualifiedValueShapesDisjoint}}{ \text{true}}$, or the empty set otherwise. Let $\nu(\vx) = \hasshape{\vx}{\vshape^{\prime}} \bigwedge_{\forall \vshape^{''} \in S^{'}} \neg \hasshape{\vx}{\vshape^{''}}$. If $\vShapeDocument$ contains the triple \\ $\trip{\vshape}{ \sh{qualifiedMinCount}}{ \conc}$, then $\alpha$ is equal to $\exists^{\geq \conc} \vy . \tau_3(\vx, \tr,\vy) \wedge \nu(\vx)$, otherwise $\alpha$ is equal to $\top$. If $\vShapeDocument$ contains the triple \\ $\trip{\vshape}{ \sh{qualifiedMaxCount}}{ \conc}$, then $\beta$ is equal to $\neg \exists^{\leq \conc} \vy . \tau_3(\vx, \tr,\vy) \wedge \nu(\vx)$, otherwise $\beta$ is equal to $\top$.
    \item
        $\tau_2(\vx, \tr, \trip{\vshape}{ \sh{close}}{ \text{true}}) \transeq \bigwedge_{\forall R \in \Theta} \neg \exists_{} \vy . R(\vx, \vy)$ if $\Theta$ is not empty, where $\Theta$ is defined as follows. Let $\Theta^{\text{all}}$ be the set of all relation names in $\vShapeDocument$, namely $\Theta^{\text{all}} = \{R | \trip{\vx}{R}{\vy} \in \vShapeDocument \}$. If this FOL translation is used to compare multiple SHACL documents, such in the case of deciding containment, then $\Theta^{\text{all}}$ must be extended to contain all the relation names in all these SHACL documents. Let $\Theta^{\text{declared}}$ be the set of all the binary property names $\Theta^{\text{declared}} = \{R | \{\trip{\vshape}{ \sh{property}}{ \vx} \wedge \trip{\vx}{\sh{path}}{R)}\} \subseteq \vShapeDocument \}$. Let $\Theta^{\text{ignored}}$ be the set of all the binary property names declared as ``ignored'' properties, namely $\Theta^{\text{ignored}} = \{R | R \in \bar{R} \wedge \allowbreak \trip{\vshape}{  \sh{ignoredProperties}}{ \bar{R}} \in \vShapeDocument\}$, where $\bar{R}$ is a list of IRIs. The set $\Theta$ can now be defined as $\Theta = \Theta^{\text{all}} \setminus (\Theta^{\text{declared}} \cup \Theta^{\text{ignored}})$. 
        
\end{itemize}

\subsection{Translation of Property Paths}

The translation $\tau_3(\vx, \tr,\vy))$ of any SHACL path $\tr$ is given by the following cases. For simplicity, we will assume that all property paths have been translated into an equivalent form having only simple IRIs within the scope of the inverse operator. Using SPARQL syntax for brevity, where the inverse operator is identified by the hat symbol $\hat \ $, the sequence path $\hat \ (r_1 / r_2)$ can be simplified into $\hat \ r_2 / \hat r_1$; an alternate path $\hat \ (r_1 \mid r_2)$ can be simplified into $\hat \ r_2 \mid \hat \ r_1$. We can simplify in a similar way zero-or-more, one-or-more and zero-or-one paths $\hat \ (r^{* / + / ?})$ into $(\hat \ r)^{* / + / ?}$.

\begin{itemize}
    \item
         If $\tr$ is an IRI $R$, then $\tau_3(\vx, \tr,\vy)) \transeq R(\vx, \vy)$
    \item
        If $\tr$ is an inverse path, with $\tr = $ ``\texttt{[ sh:inversePath $R$ ]}'', then $\tau_3(\vx, \tr,\vy)) \transeq 
        R^{-}(\vx, \vy) $  
    \item
        If $\tr$ is a conjunction of paths, with $\tr = $ ``\texttt{( $\tr_1$, $\tr_2$, ..., $\tr_n$ )}'', then $\tau_3(\vx, \tr,\vy)) \transeq 
        \exists \vz_1, \vz_2, ..., \vz_{n-1} . \tau_3(\vx, \tr_1,\vz_1)) \wedge \tau_3(\vz_1, \tr_2,\vz_2)) \wedge ... \wedge \tau_3(\vz_{n-1}, \tr_2,\vy)) $
    \item
        If $\tr$ is a disjunction of paths, with $\tr = $ ``\texttt{[ sh:alternativePath ( $\tr_1$, $\tr_2$, ..., $\tr_n$ ) ]}'', then $\tau_3(\vx, \tr,\vy)) \transeq 
        \tau_3(\vx, \tr_1,\vy)) \vee \tau_3(\vx, \tr_2,\vy)) \vee ... \vee \tau_3(\vx, \tr_n,\vy)) $   
    \item
        If $\tr$ is a zero-or-more path, with $\tr = $ ``\texttt{[ \sh{zeroOrMorePath} $\tr_1$]}'', then $\tau_3(\vx, \tr,\vy)) \\ \transeq 
        (\tau_3(\vx, \tr_1,\vy)))^{*}$
    \item
        If $\tr$ is a one-or-more path, with $\tr = $ ``\texttt{[ \sh{oneOrMorePath} $\tr_1$]}'', then $\tau_3(\vx, \tr,\vy)) \transeq \exists \vz . \tau_3(\vx, \tr_1,\vz)) \wedge (\tau_3(\vz, \tr_1,\vy)))^{*}$ 
    \item
        If $\tr$ is a zero-or-one path, with $\tr = $ ``\texttt{[ \sh{zeroOrOnePath} $\tr_1$]}'', then $\tau_3(\vx, \tr,\vy)) \transeq \vx = \vy \vee \tau_3(\vx, \tr_1,\vy))$  
\end{itemize}

%% file: zappendixFOL2SHACL.tex
\section{Translation from \SCL\ Grammar into SHACL}

We present here the translation $\mu$, inverse of $\tau$, to translate a sentence in the \SCL\ grammar into a SHACL document. We begin by defining the translation of the property path subgrammar $r(\vx, \vy)$ into SHACL property paths:
\begin{itemize}
    \item $\mu(R) \transeq R$
    
    \item $\mu(R^{-}) \transeq $ \texttt{ [ \sh{inversePath} $R$ ] }
    
    \item $\mu(r^{\star}(\vx, \vy)) \transeq $ \\ \texttt{ [ \sh{zeroOrMorePath} $\mu(r(\vx, \vy))$ ] } 
    
    \item $\mu(\vx = \vy \vee r(\vx, \vy)) \transeq $ \\ \texttt{ [ \sh{zeroOrOnePath} $\mu(r(\vx, \vy))$ ] } %
    
    \item $\mu(r_1(\vx, \vy) \vee r_2(\vx, \vy)) \transeq $ \\ \texttt{ [ \sh{alternativePath} ( $\mu(r_1(\vx, \vy))$, $\mu(r_2(\vx, \vy))$ ] ) } 
    
    \item $\mu(r_1(\vx, \vy) \wedge r_2(\vx, \vy)) \transeq $ \\ \texttt{ ( $\mu(r_1(\vx, \vy))$, $\mu(r_2(\vx, \vy))$ ) }
\end{itemize}

The translation of the constraint  subgrammar $\psi(\vx)$ is the following. we will use $\mu(\psi(\vx))$ to denote the SHACL translation of shape $\psi(\vx)$, and $\iota(\mu(\psi(\vx)))$ to denote its shape IRI. To improve legibility, we omit set brackets around sets of RDF triples, and we represent them in Turtle syntax. For example, a set of RDF triples such as ``\texttt{$s$ a \sh{NodeShape} ; \sh{hasValue} \conc\ .} '' is to be interpreted as the set $\{\trip{s}{\rdf{type}}{\sh{NodeShape}}, \trip{s}{\sh{hasValue}}{c}\}$.


\begin{itemize}
    \item $\mu(\top) \transeq  $ \\ \texttt{$s$ a \sh{NodeShape} . } 
    
    \item $\mu(\vx = \conc) \transeq  $ \\ \texttt{$s$ a \sh{NodeShape} ; \\
    \phantom{x} \sh{hasValue} \conc\ . 
    } 
    
    \item $\mu(F(\vx)) \transeq $ \\ \texttt{$s$ a \sh{NodeShape} ; \\
    \phantom{x} $f$ \conc\ . 
    } \\ Predicate $f$ is the filter function identified by $F$, namely one of the following: \sh{datatype}, \sh{nodeKind}, \sh{minExclusive}, \sh{minInclusive}, \\ \sh{maxExclusive}, \sh{maxInclusive}, \sh{maxLength}, \sh{minLength}, \sh{pattern}, \sh{languageIn}.
    
    \item $\mu(\hasshape{\vx}{\vshape^{\prime}}) \transeq  $ \\ \texttt{$s$ a \sh{NodeShape} ; \\
    \phantom{x} \sh{node} $\vshape^{\prime}$ . 
    } \\ if  $\vshape^{\prime}$ is a node shape, else:\\
    \texttt{$s$ a \sh{NodeShape} ; \\
    \phantom{x} \sh{property} $\vshape^{\prime}$ . 
    }
    
    \item $\mu(\neg \psi(\vx)) \transeq  $ \\ \texttt{$s$ a \sh{NodeShape} ; \\
    \phantom{x} \sh{not} $\iota(\mu(\psi(\vx)))$ . } 
    
    \item $\mu(\psi_1(\vx) \wedge \psi_2(\vx)) \transeq  $ \\ \texttt{$s$ a \sh{NodeShape} ; \\
    \phantom{x} \sh{and} ($\iota(\mu(\psi_1(\vx)))$, $\iota(\mu(\psi_2(\vx)))$) . }
    
    \item $\mu(\exists^{\geq n} \vy . r(\vx, \vy) \wedge \psi(\vx)) \transeq  $ \\ \texttt{$s$ a \sh{NodeShape} ; \\
    \phantom{x} \sh{property} [ \\
    \phantom{x} \phantom{x} \sh{path} $\mu(r(\vx, \vy))$ ; \\ 
    \phantom{x} \phantom{x} \sh{qualifiedValueShape} $\iota(\mu(\psi(\vx)))$ ; \\
    \phantom{x} \phantom{x} \sh{qualifiedMinCount} $n$ ;     \\
    \phantom{x} ] .   
    }
    
    \item $\mu( \forall \vy . r(\vx, \vy) \leftrightarrow R(\vx, \vy) ) \transeq  $ \\ \texttt{$s$ a \sh{NodeShape} ; \\
    \phantom{x} \sh{property} [ ; \\
    \phantom{x} \phantom{x} \sh{path} $\mu(r(\vx, \vy))$ ; \\ 
    \phantom{x} \phantom{x} \sh{equals} $R$ ;     \\
    \phantom{x} ] .   
    }
    
    \item $\mu( \neg \exists \vy . r(\vx, \vy) \wedge R(\vx, \vy) ) \transeq  $ \\ \texttt{$s$ a \sh{NodeShape} ; \\
    \phantom{x} \sh{property} [ ; \\
    \phantom{x} \phantom{x} \sh{path} $\mu(r(\vx, \vy))$ ; \\ 
    \phantom{x} \phantom{x} \sh{disjoint} $R$ ;     \\
    \phantom{x} ] .   
    }    
    
    \item $\mu( \forall \vy, \vz . r(\vx, \vy) \wedge R(\vx, \vz) \rightarrow \vy < \vz ) \transeq  $ \\ \texttt{$s$ a \sh{NodeShape} ; \\
    \phantom{x} \sh{property} [ ; \\
    \phantom{x} \phantom{x} \sh{path} $\mu(r(\vx, \vy))$ ; \\ 
    \phantom{x} \phantom{x} \sh{lessThan} $R$ ;     \\
    \phantom{x} ] .   
    }   

    \item $\mu( \forall \vy, \vz . r(\vx, \vy) \wedge R(\vx, \vz) \rightarrow \vy \leq \vz ) \transeq  $ \\ \texttt{$s$ a \sh{NodeShape} ; \\
    \phantom{x} \sh{property} [ ; \\
    \phantom{x} \phantom{x} \sh{path} $\mu(r(\vx, \vy))$ ; \\ 
    \phantom{x} \phantom{x} \sh{lessThanOrEquals} $R$ ;     \\
    \phantom{x} ] .   
    }   
    
\end{itemize}

We can now define the translation $\mu(\varphi)$ of a complete sentence of the $\varphi$-grammar into a SHACL document \vShapeDocument\ (effectively a set of RDF triples) as follows.  

\begin{itemize}
    \item $\mu(\varphi_1 \wedge \varphi_2) \transeq  \mu(\varphi_1) \cup \mu(\varphi_2)) $ 
    
    \item $\mu(\psi_1(\conc)) \transeq  \mu(\psi_1(\vx)) \; \cup $ \\ \texttt{$s$ a \sh{NodeShape} ; \\
    \phantom{x} \sh{targetNode} \conc ; \\
    \phantom{x} \sh{node} $\iota(\mu(\psi_1(\vx)))$ . 
    }
    
    \item $\mu(\forall \vx . \; \isA(\vx, \conc) \rightarrow \psi_1(\vx)) \transeq  \mu(\psi_1(\vx)) \; \cup $ \\ \texttt{$s$ a \sh{NodeShape} ; \\
    \phantom{x} \sh{targetClass} \conc ; \\
    \phantom{x} \sh{node} $\iota(\mu(\psi_1(\vx)))$ .
    } 

    \item $\mu(\forall \vx, \vy . \; R(\vx, \vy) \rightarrow \psi_1(\vx)) \transeq  \mu(\psi_1(\vx)) \; \cup $ \\ \texttt{$s$ a \sh{NodeShape} ; \\
    \phantom{x} \sh{targetSubjectsOf} $R$ ; \\
    \phantom{x} \sh{node} $\iota(\mu(\psi_1(\vx)))$ .
    }     

    \item $\mu(\forall \vx, \vy . \; R^{-}(\vx, \vy) \rightarrow \psi_1(\vx)) \transeq  \mu(\psi_1(\vx)) \; \cup $ \\ \texttt{$s$ a \sh{NodeShape} ; \\
    \phantom{x} \sh{targetObjectsOf} $R$ ; \\
    \phantom{x} \sh{node} $\iota(\mu(\psi_1(\vx)))$ .
    }     
    
    \item $\mu(\forall \vx . \; \hasshape{\vx}{\sconst} \leftrightarrow \psi(\vx)) \transeq  \mu(\psi_1(\vx)) \; \cup $ \\ 
    \texttt{$s$ a \sh{NodeShape} ; \\
    \phantom{x} \sh{node} $\iota(\mu(\psi_1(\vx)))$ .
    }   
    
\end{itemize}

%% file: zappendixProof.tex
\section{Additional Proof}

\begin{proof}[Proof of Theorem~\ref{thm:fmp(oeop)}]
Similarly to the use of the \C\ construct of \SCL, a simple combination of
few instances of the \lO\ feature allows to write the following sentence
$\varphi$ encoding the existence of an injective function that is not
surjective.
In more detail, a weaker version of the role of the counting quantifier is
played here by the \lO' construct that enforces the functionality of the two
relations $F$ and $G$.
In addition, by applying the \lO\ construct twice between the inverse of $F$
and $G$, we are able to ensures that $F^{-}$ is functional as well.
Hence, the thesis easily follows.
\[
\begin{aligned}
  \varphi \grameq\;
    & \isA(0, \conc) \wedge \neg \exists \vx \,.\, F^{-}(0, \vx) \wedge
      \forall \vx \,.\, \isA(\vx, \conc) \rightarrow \psi(\vx); \\
  \psi(\vx) \grameq\;
    & \exists \vy \,.\, (F(\vx, \vy) \wedge \isA(\vy, \conc)) \\
  \wedge\;
    & \forall \vy, \vz \,.\, F(\vx, \vy) \wedge F(\vx, \vz) \rightarrow
      \vy \leq \vz \wedge \forall \vy, \vz \,.\, G(\vx, \vy) \wedge
      G(\vx, \vz) \rightarrow \vy \leq \vz \\
  \wedge\;
    & \forall \vy, \vz \,.\, F^{-}(\vx, \vy) \wedge G(\vx, \vz) \rightarrow
      \vy \leq \vz \wedge \forall \vy, \vz \,.\, F^{-}(\vx, \vy) \wedge
      G(\vx, \vz) \rightarrow \vz \leq \vy.
\end{aligned}
\]

To prove that \E\,\lO' fragment does not enjoy the finite-model property
too, it is enough to replace the last to applications of the \lO\ feature
with the \E\ formula $\forall \vy \,.\, F^{-}(\vx, \vy) \leftrightarrow
G(\vx, \vy)$, which ensures the functionality of $F^{-}$, being $G$
functional.
\end{proof}